\documentclass[12pt]{article}
\usepackage{amsmath}
\usepackage{graphicx,psfrag,epsf}
\usepackage{enumerate}
\usepackage{color, colortbl}
\usepackage{url} 
\usepackage{float}

\usepackage{caption}
\usepackage{subcaption}

\newcommand{\blind}{0}

\usepackage{amsthm,amsfonts,amssymb}
\usepackage{comment}
\usepackage{lscape}
\usepackage[backend=bibtex, style=nature, citestyle=authoryear]{biblatex}
\bibliography{bibliography}

\def\inprob{\stackrel{p}{\rightarrow}}
\def\indist{\rightsquigarrow}
\def\ind{\perp\!\!\!\perp}

\newcommand{\var}{\text{var}}

\newcommand{\Pb}{\mathbb{P}}
\newcommand{\Pn}{\mathbb{P}_n}

\newcommand{\E}{\mathbb{E}}
\newcommand{\R}{\mathbb{R}}

\definecolor{Gray}{gray}{0.9}

\newcommand{\bX}{\mathbf{X}}
\newcommand{\bx}{\mathbf{x}}

\newcommand{\bO}{\mathbf{O}}

\newcommand{\bdelta}{\boldsymbol\delta}

\allowdisplaybreaks

\DeclareSymbolFont{bbold}{U}{bbold}{m}{n}
\DeclareSymbolFontAlphabet{\mathbbold}{bbold}
\newcommand{\one}{\mathbbold{1}}
\newtheorem{theorem}{Theorem}
\newtheorem{lemma}{Lemma}

\newtheorem{proposition}{Proposition}

\theoremstyle{remark}
\newtheorem{assumption}{Assumption}

\newtheorem{condition}{Condition}

\begin{document}

\def\spacingset#1{\renewcommand{\baselinestretch}%
{#1}\small\normalsize} \spacingset{1}


\if0\blind
{
  \title{ \bf Survivor-complier effects in the presence of selection on treatment, with application to a study of prompt ICU admission  }
  
  \author{Edward H. Kennedy \thanks{Assistant Professor, Department of Statistics, Carnegie Mellon University (edward@stat.cmu.edu).} \\ Carnegie Mellon University
    \and
    Steve Harris \thanks{Clinical Lecturer in Anaesthesia and Critical Care,  University College London.} \\
    University College London Hospital
    \and
    Luke J. Keele \thanks{Associate Professor,  McCourt School of Public Policy, Georgetown University.} \\
    Georgetown University 
     }
  \maketitle
  \thispagestyle{empty}
  \setcounter{page}{0}
} \fi

\if1\blind
{
  \bigskip
  \bigskip
  \bigskip
  \begin{center}
    {\LARGE\bf Survivor-complier effects in the presence of selection on treatment, with application to a study of prompt ICU admission  }
\end{center}
  \setcounter{page}{0}
  \medskip
} \fi

\begin{abstract}
Pre-treatment selection or censoring (`selection on treatment') can occur when two treatment levels are compared ignoring the third option of neither treatment, in `censoring by death' settings where treatment is only defined for those who survive long enough to receive it, or in general in studies where the treatment is only defined for a subset of the population. Unfortunately, the standard instrumental variable (IV) estimand is not defined in the presence of such selection, so we consider estimating a new survivor-complier causal effect. Although this effect is generally not identified under standard IV assumptions, it is possible to construct sharp bounds. We derive these bounds and give a corresponding data-driven sensitivity analysis, along with  nonparametric yet efficient estimation methods. Importantly, our approach allows for high-dimensional confounding adjustment, and valid inference even after employing machine learning. Incorporating covariates can tighten bounds dramatically, especially when they are strong predictors of the selection process. We apply the methods in a UK cohort study of critical care patients to examine the mortality effects of prompt admission to the intensive care unit, using ICU bed availability as an instrument. 
\end{abstract}

\noindent%
\vfill

\thispagestyle{empty}

\newpage

\spacingset{1.45} 
\spacingset{1}

\section{Introduction}

\subsection{Time to ICU Admission}

In UK hospitals, nurses can call for a critical care assessment to judge whether deteriorating patients on general hospital wards should be transferred to the intensive care unit (ICU) for a higher level of critical care. When a critical care assessment is ordered, staff from the ICU assess whether the patient would improve under ICU care. Patients can at this time be admitted to the ICU, but they may not be directly transferred to the ICU. If ICU bed space is limited, the patient is admitted but often must wait to be transferred to the ICU. Clinical guidelines in the UK indicate that a patient should be transferred to the ICU within four hours \autocite{icu2013}.

One open clinical question is whether delayed admission to the ICU is harmful. More specifically: does admission to the ICU in less than four hours contribute to decreased mortality? This question is made more complex by two factors. First, time to admission is undefined for patients that are not admitted to the ICU at the time of assessment. That is, we do not observe a time to transfer for any patient that is not admitted. As such, important information is censored for part of the study population. Second, the mix of patients makes naive comparisons of mortality rates for patients quickly transferred to the ICU to those with delayed admission deceiving. It is very likely that patients with short wait times tend to be sicker than patients with delayed transfer. If that is the case, prompt transfer to the ICU would be associated with higher not lower mortality rates.  However, this association arises because of the characteristics of those patients who are promptly transferred, rather than because of their prompt ICU care.

A carefully designed randomized experiment could overcome these difficulties. If we were to randomize all admitted to patients to different wait times for admission, we could identify the causal effect among a specific sub-population. However, ethical constraints make an experiment of this type infeasible. Alternatively we might seek to find a natural experiment. A natural experiment is some naturally occurring circumstance that results in haphazard if not as-if random assignment of the treatment. An instrument is one type of natural experiment. In our case, an instrumental variable (IV) would give a haphazard nudge towards being admitted to the ICU, but would only affect outcomes indirectly through ICU admission and by no other means. 

However, standard IV methods only identify the effect of being admitted to the ICU, while we seek to identify the causal effect of delayed transfer to the ICU. Since wait times are undefined or censored for those patients not admitted to the ICU, standard IV methods fail. Therefore we develop methods for estimating the effect of delayed transfer, among the principal stratum of patients who would always be admitted to the ICU regardless of bed availability. Unfortunately this effect is not point-identified under standard IV assumptions, but we can nonetheless construct sharp covariate-adjusted bounds and corresponding flexible estimation methods. Combining IV methods with partial identification approaches has a long history, to which our work contributes \autocite{manski1990nonparametric, manski1996learning, balke1997bounds, kitagawa2009identification, siddique2013partially, MeaPac2013}. Next, we describe the data in greater detail.

\subsection{Data \& Descriptive Statistics}

The data we use is from a prospective cohort study of general ward patients who were referred to critical care in the ICU in 48 National Health Service (NHS) hospitals between 1 November 2010 and 31 December 2011 (the SPOTlight study) \autocite{harris2015delay}. The data record information for every general ward patient that was assessed for possible admission to the ICU including the decision to admit a patient to ICU care and the time from assessment to actual transfer into the ICU. Importantly, data collection also included a measure of ICU occupancy rates at the specific time of the patient assessment. This measure of ICU bed availability at time of assessment serves as an instrument in our study.  The data collection also included a number of baseline covariates such as age, septic diagnosis (0/1), peri-arrest (0/1), and measures of physiology. These measures of physiology include the Intensive Care National Audit \& Research Centre (ICNARC) physiology score, the NHS National Early Warning Score (NEWS) which measures whether respiratory rate, oxygen saturations, temperature, systolic blood pressure, pulse rate, a level of consciousness vary from the norm, and the Sequential Organ Failure Assessment (SOFA) score which ranges from 0 to 24, with higher scores indicating a greater degree of organ failure. The data also record the patient's existing level of care at assessment and recommended level of care after assessment using the UK Critical Care Minimum Dataset (CCMDS) levels of care. These levels are 0 and 1 for normal ward care, 2 for care within a high dependency unit, and 3 for care with intensive care unit. Finally, the data include indicators for whether it was the weekend, out of hours (between 7 PM and 7 AM), or the months from November to February.  The primary endpoint is 28-day mortality. 

In total, the data contain information on 15158 patients on general hospital wards that were assessed for admission to the ICU. From this total, we exclude 2141 patients due to the presence of a treatment limitation order which excluded the possibility of transfer to the ICU.  We further exclude six additional patients that have missing data on the availability of beds in the ICU at the time of assessment, giving us a study population of 13011 patients that were assessed for possible transfer to the ICU. Of these 13011 patients, 27\% were admitted to the ICU with a median time to transfer of 2 hours. However, 8.3\% of patients that were initially admitted were never transferred to the ICU. Figure~\ref{fig:icu_wait} plots the distribution of wait times for transfer to the ICU, where it is clear that while many patients are transferred within four hours, the time for many patients significantly exceeds that threshold.

\begin{figure}
\centering
  \includegraphics[scale=.5]{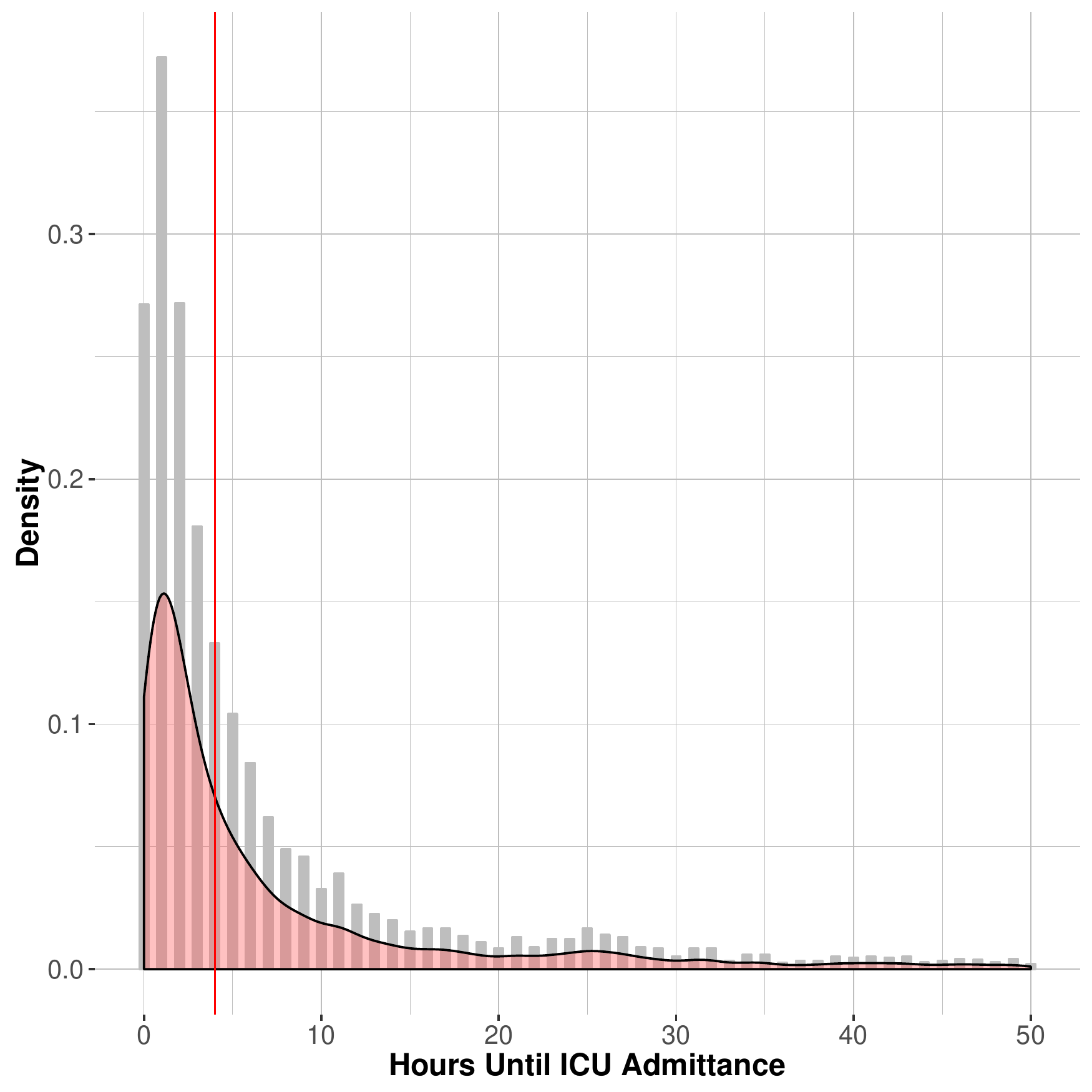}
\caption{Distribution of ICU wait times. Vertical red line indicates four hour mark. Three observations with wait times in excess of 50 hours removed from the data.}
  \label{fig:icu_wait}
\end{figure}

\subsection{An Instrumental Variables Approach}

As we noted above, we address the question of ICU wait times within an IV framework. Specifically, we use ICU bed availability as an instrument for admission to the ICU. That is, if many beds are available at the time of assessment, this should serve as a haphazard encouragement for ICU admittance. Under a set of causal identification assumptions, IVs allow for the identification of causal effects subject to a form of unobserved confounding that is likely in our application \parencite{angrist1996identification}. See \textcite{keele2016stronger} for an in-depth study of bed availability as an instrument for ICU admission.  

Any use of IVs requires careful assessment of whether the identification assumptions are plausible. However, standard application of IV methods in our application will fail due to selection bias.  As we noted above, the waiting time for transfer to the ICU is only observed when a patient is admitted to the ICU. Thus for patients that are assessed but not admitted to critical care, wait times and the counterfactual mortality outcome under wait times are undefined. In our application, selection arises through this censoring, since the treatment may be undefined due to selection that occurs prior to treatment. 

Some recent papers haved considered settings where the outcome can be undefined but treatment is always observed, i.e., it is assumed that the selection or censoring event (e.g., death) can only occur \textit{after} treatment \parencite{imai2008sharp,yang2016using}. However, this is substantially different from our application, where the selection is based on treatment status; as we show, selection on treatment leads to different causal estimands and requires new methods. More relevant is the work by \textcite{swanson2015selecting,ertefaie2016selection}, who first discovered and discussed the fundamental selection-on-treatment problem with IVs. Our work adds to these papers in several critical respects. First, \textcite{swanson2015selecting} importantly pointed out the problem and discussed the biases of standard IV analyses, but did not present any solutions. \textcite{ertefaie2016selection} followed with a sensitivity analysis approach for the same estimand we consider in this paper, but their approach requires numerous sensitivity parameters, does not yield provably sharp bounds, and does not accommodate covariates (even low-dimensional). In contrast, our approach can exploit potentially high-dimensional covariate information to make bounds as sharp as possible (and weaken identifying assumptions), and otherwise is indexed by only two interpretable sensitivity parameters. Importantly we use influence function theory to develop nonparametric estimators that can converge at fast parametric rates, even when based on flexible machine learning methods; our approach can similarly be extended to other partial identification problems involving high-dimensional covariates, which we feel is a critical but understudied area.

This article is organized as follows. In Section~\ref{sec:note} we outline notation, define the causal estimand, and derive identification conditions. In Section~\ref{sec:methods} we outline a set of sharp nonparametric bounds for the causal estimand of interest. We also derive a data-driven sensitivity analysis that uses the bounds to allow for more informative inferences, and propose new estimators that allow for high-dimensional covariate adjustment. We then apply the methods in the aforementioned UK cohort study to examine mortality effects of prompt admission to the ICU. In Section~\ref{sec:sim} we use simulations to explore the conditions under which the bounds are most informative. Section~\ref{sec:dis} concludes.

\section{Notation, Causal Estimand, \& Identification}
\label{sec:note}

\subsection{Notation}

We suppose we observe an independent and identically distributed sample $(\bO_1, ..., \bO_n)$ from distribution $\Pb$ where $\bO=(\bX, Z, S, A, Y)$ with $\bX$ denoting covariate information, $Z$ a binary instrument, $S$ an indicator of the selection event, $A$ a binary treatment that is only defined when $S=1$, and $Y$ some outcome of interest (which may also only be defined when $S=1$). In our example, $Z$ is an indicator of bed availability (with $Z=1$ indicating that there are more than the average number of beds available), and $S$ indicates whether a patient was accepted for ICU admission. Treatment $A=1$ indicates a wait time less than 4 hours before actually being transferred into the ICU. For all patients not admitted to the ICU (i.e., with $S=0$), wait time until transfer into the ICU is undefined, as are counterfactual mortality outcomes under different waiting times.  In our application, $Y$ is a binary indicator of death; however all results can also be extended to arbitrary bounded outcomes.

A directed acyclic graph showing the causal structure is given in Figure~\ref{fig:dag}, with arrows between variables indicating causal relationships. The graph shows baseline covariates $\bX$, instrument $Z$, selection event $S$ (the box enclosing $S$ indicates that subsequent variables are only defined conditional on $S=1$), treatment $A$, and outcome $Y$, along with any unmeasured confounders $U$. If only time-ordering is taken into account, then all arrows would be present (excluding those that allow the future to cause the past); however the relationships underlying the gray dotted arrows are assumed absent by standard IV assumptions, as discussed in the next subsection.

\begin{figure}[ht]
\spacingset{1}
\begin{center}
\includegraphics[width=3.7in]{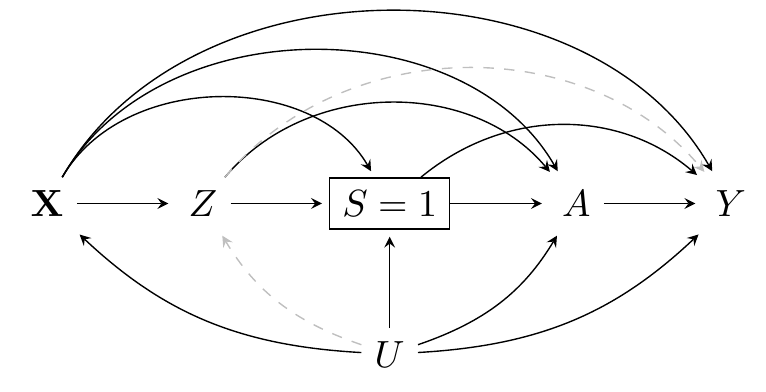}
\end{center}
\caption{Directed acyclic graph showing covariates $\bX$, instrument $Z$, selection event $S$, treatment or exposure $A$, outcome $Y$, and unmeasured variables $U$. The exposure $A$ (and possibly outcome $Y$) are only defined when $S=1$. Gray dotted arrows indicate relationships that are assumed absent in a standard instrumental variable analysis. \label{fig:dag}}
\end{figure}

Although we illustrate the data structure in graphical terms in Figure~\ref{fig:dag}, we use potential outcome notation to formalize assumptions and characterize causal effects \autocite{rubin1974estimating}. Specifically we let $(S^{(z)}, A^{(z)}, Y^{(z)})$ denote selection, treatment, and outcome had the instrument been set to $Z=z$, and similarly we let $(Y^{(z,a)},Y^{(a)})$ denote outcomes had the instrument and/or treatment been set to $Z=z$ and/or $A=a$. Our goal is to learn about the distribution of the effect  $Y^{(a=1)}-Y^{(a=0)}$, but this contrast only makes sense for some units in the population, as we will make more explicit in the next subsection.

Finally, we note that selection on treatment can also arise when the goal is to compare only two levels of a multivalued treatment \autocite{swanson2015selecting, ertefaie2016selection}. In this setting, $S=1$ would indicate that one of the two treatment levels of interest was received (e.g., statin type), while $S=0$ would indicate receipt of an alternative treatment (e.g., no statins at all); here $A$ would indicate which of the two treatments of interest was received, and would be undefined when $S=0$. This kind of selection bias also occurs when patients die during follow-up, after encouragement by the instrument but prior to the measurement of treatment. Then the selection event $S$ is survival, so that $S=0$ means a patient died; in this case subsequent treatment and outcomes are undefined, and it often does not make sense to consider what values might have been had the patient survived \autocite{tchetgen2014identification}.

\subsection{Causal Estimand}

In this subsection we make explicit our target of inference, i.e., formally define which causal contrast we aim to identify and estimate. Since the treatment of interest $A$ is only defined when $S=1$ (e.g., how long a patient waits to be transferred to the ICU is only defined for patients who are admitted to be transferred to the ICU), and since setting $S=1$ for all patients might not be a meaningful counterfactual (e.g., some patients might be too healthy or too sick to ever be admitted to the ICU), it is reasonable to restrict inference to the subgroup of patients who would always be selected, regardless of instrument value, i.e., those patients with
$$ S^{(z=1)}=S^{(z=0)}=1  $$
with probability one. These are the patients who would always be admitted to the ICU, regardless of bed availability. Note that this is really the only subgroup for which we can discuss treatment effects, since treatment is not defined if $S=0$. This is an example of a principal stratum since it is defined based on a cross-classification of joint counterfactuals \autocite{frangakis2002principal}. This particular kind of principal stratum commonly arises in censoring-by-death settings \autocite{zhang2003estimation, tchetgen2014identification} where counterfactuals are defined based on treatment rather than instrument values, as well as in vaccine effect studies where interest centers on effects among sicker patients who would always be infected regardless of treatment status \autocite{hudgens2003analysis}.


One could also argue that the distribution of $Y^{(a=1)}-Y^{(a=0)}$ should be explored in a further refined subgroup, in particular only within those always-selected or survivor patients who would actually respond to encouragement by the instrument, i.e., those with
$$ A^{(z=1)} > A^{(z=0)} . $$
These units are commonly termed compliers in the IV literature \autocite{angrist1996identification}. Under appropriate assumptions, these are patients whose wait time would depend on bed availability. That is, these patients would only wait less than 4 hours if there were beds availabile at the time of critical care assessment. There is an ongoing debate about the merits of exploring causal effects within such subgroups \autocite{imbens2014instrumental,swanson2014think}. In this paper we take a pluralistic perspective, viewing complier effects as meaningful but only one piece of the puzzle, which can be combined with inference about other subgroups. Our focus on them is driven by both subject matter concerns in our ICU application and because they are a popular target of inference in IV studies.

As such, in our analysis we focus on the causal effect among those always-selected patients who would respond to IV encouragement, i.e., 
$$ \psi = \E\Big( Y^{(a=1)} - Y^{(a=0)} \Bigm| S^{(z=1)}=S^{(z=0)}=1, A^{(z=1)}>A^{(z=0)} \Big) . $$
We call this parameter the survivor-complier average treatment effect, or SCATE. It represents the mortality risk difference due to waiting less than four hours ($A=1$) versus more than four hours ($A=0$), among patients who would always be admitted to the ICU but whose wait time would depend on bed availability. Negative values of this esitmand indicate that shorter wait times reduce mortality in this subgroup.

\subsection{Identification Conditions}

Next, we summarize standard IV assumptions that we rely on in this paper, and discuss why the SCATE $\psi$ is not point-identified under them. We also give an expression for $\psi$ that will be important for constructing bounds and interpretable sensitivity analyses. Since $\psi$ is expressed in terms of only partially observed potential outcomes, we will need causal assumptions to be able to learn anything about it. Instead of tailoring our assumptions so as to ensure point identification (i.e., so that $\psi$ can be uniquely expressed as a functional of the observed data distribution $\Pb$), we instead see what we can learn about the SCATE $\psi$ under standard IV assumptions. Specifically we assume the following:

\begin{assumption}[Consistency]
\begin{align*}
S &= ZS^{(z=1)} + (1-Z)S^{(z=0)}, \\
A &= ZA^{(z=1)} + (1-Z)A^{(z=0)} \ \text{(if $S=1$)}, \\
Y &= ZY^{(z=1)} + (1-Z)Y^{(z=0)} \ \text{(if $S=1$)}. 
\end{align*}
\end{assumption}
\begin{assumption}[Positivity]
$0<\Pb(Z=1 \mid \bX) < 1$.
\end{assumption}
\begin{assumption}[Instrumentation]
$\Pb(A^{(z=1)}>A^{(z=0)} , S^{(z=1)}=S^{(z=0)}=1) \geq \epsilon > 0$.
\end{assumption}
\begin{assumption}[Unconfoundedness]
$Z \ind ( S^{(z)}, A^{(z)}, Y^{(z)}) \mid \bX$ for $z=0,1$.
\end{assumption}
\begin{assumption}[Exclusion]
$\E(Y^{(z,a)} - Y^{(a)} \mid S^{(z=1)} = S^{(z=0)}=1)=0$ for $z=0,1$.
\end{assumption}
\begin{assumption}[Monotonicity Restriction]
\begin{equation*}
\begin{gathered}
\Pb(S^{(z=1)} \geq S^{(z=0)}) = 1 , \\
\Pb(A^{(z=1)} \geq A^{(z=0)} \mid S^{(z=1)} = S^{(z=0)}=1) = 1 .
\end{gathered}
\end{equation*}
\end{assumption}

Analogs of Assumptions 1--6 are all standard in the IV literature, although monotonicity is sometimes replaced by effect homogeneity or no-interaction assumptions \autocite{robins1994correcting, hernan2006instruments,tan2010marginal}. Under consistency, we assume that any variation in how patients receive care in the ICU corresponds to the same potential outcomes. The positivity assumption implies that, regardless of covariate value, each patient has some chance of receiving each instrument value. Under instrumentation, the survivor-complier subgroup must be nonempty. We can partly assess this assumption from the data.  When the number of ICU beds available is higher than average, patients are admitted to the ICU 61\% of the time, and they waited less than four hours 42\% of the time. When the number of beds available was lower than average, patients were admitted to the ICU 53\% of the time but only 34\% waited less than four hours. Unconfoundedness says that (within strata of observed covariates) the instrument is as good as randomized, since it is conditionally independent of potential outcomes. In the data, bed availability in general tells us little about patient characteristics or risk severity. If the exclusion restriction holds among the always-selected, the instrument has no direct effect on outcomes and thus can only affect outcomes through the treatment.  We judge the exclusion restriction plausible, since bed occupancy in the ICU is measured when the patient is assessed for ICU admission. To violate the exclusion restriction, there would have to be some aspects of bed availability at the time of assessment that contributes directly to patient mortality. Finally, the monotonicity assumption implies that there are no units who would always do the opposite of what the instrument encourages them to do, either with respect to selection or treatment; in other words this assumption rules out the possibility of `defiers'. The monotonicity assumption would be violated if the hospital staff assessing patients for ICU admission encouraged prompt transfer to the ICU when there were few beds available, and discouraged prompt transfer when there were many beds available. Such behavior seems unlikely. See \autocite{angrist1996identification, hernan2006instruments} for a more in-depth discussion of Assumptions 1--6, including ways in which they could be violated.

With $(Z,S,A)$ all binary, there are nine possible principal strata, three of which are ruled out by monotonicity; this is illustrated in Table~\ref{tab1}. Monotonicity implies that any patient who would be admitted with few beds available would also be admitted with more beds available, and among those patients who would always be admitted to the ICU, having more beds available would not delay transfer times relative to having fewer beds available. Thus, in this application, monotonicity is probably a reasonable assumption.  

\begin{table}[h!] 
\spacingset{1}
\caption{Principal strata defined by $(Z,S,A)$ and strata excluded by assumption. Strata excluded by monotonicity are shaded in gray. Primary stratum of interest is in row seven.} 
\label{tab1}
\begin{center}
\begin{tabular}{rrrr}
$S^{(z=0)}$ & $S^{(z=1)}$ & $A^{(z=0)}$ & $A^{(z=1)}$ \\ 
\hline
0 & 0 & $\cdot$ & $\cdot$ \\
0 & 1 & $\cdot$ & 0 \\
0 & 1 & $\cdot$ & 1 \\
\rowcolor{Gray} 1 & 0 & 0 & $\cdot$ \\
\rowcolor{Gray}  1 & 0 & 1 & $\cdot$ \\
1 & 1 & 0 & 0 \\
1 & 1 & 0 & 1 \\
\rowcolor{Gray} 1 & 1 & 1 & 0 \\
1 & 1 & 1 & 1 \\
\end{tabular}
\end{center}
\end{table}

Under Assumptions 1--6 the SCATE parameter $\psi$ is not identified. The following proposition states this, and is proved in the Appendix (with all other results).

\begin{proposition}
The survivor-complier effect $\psi$ is not identified under Assumptions 1--6, i.e., there exist data-generating processes with $\psi_1 \neq \psi_2$ that imply the same observed probability distribution $\Pb$ for $(\bX,Z,S,A,Y)$.
\end{proposition} 

In other words, because the same observed data distribution $\Pb$ can be consistent with different values of $\psi$, even knowing $\Pb$ without error (e.g., with infinite sample size) is not enough to learn the true value of $\psi$ with certainty. As pointed out for example by \textcite{yang2016using, ertefaie2016selection} in related identification problems, this is essentially a result of the fact that the conditional distribution of $Y$ given $(Z,S,A)$ is a mixture of potential outcome distributions from different principal strata. 

Although the SCATE $\psi$ is not point-identified, we show in the next section that it is still possible to construct informative bounds and sensitivity analyses. Before that, though, we first define a fundamental expression for the $\psi$, indicating how it relates to the proportion of survivor-compliers and an intention-to-treat effect. Again, proofs of all results are given in the Appendix.

\begin{proposition}
Under Assumptions 1, 3, 5, and 6, we can write the SCATE as
$$ \psi = (\beta / \alpha) \Pb(S^{(z=1)} = S^{(z=0)}=1) $$
where
$$ \alpha = \Pb(S^{(z=1)}=S^{(z=0)}=1, A^{(z=1)}>A^{(z=0)}) $$
is the survivor-complier proportion, and
$$ \beta = \E(Y^{(z=1)} - Y^{(z=0)} \mid S^{(z=1)} = S^{(z=0)}=1) $$
is the survivor intention-to-treat effect (i.e., effect of $Z$ among always-selected).
\end{proposition}

The parameter $\alpha$ is the survivor-complier proportion, which is a measure of the strength of the instrument, and is also important since it measures the size of the non-identified population in which the effect is being estimated. In our application, this is the proportion of patients that wait less than four hours to gain access to the ICU because they were encouraged by bed availability in the ICU. The intention-to-treat effect, $\beta$, is the change in mortality caused by bed availability among the sub-population that would always be admitted to the ICU in less than four hours. The term $\beta$, in addition to providing an upper-bound on complier effects, does not require monotonicity of treatment (though in our setting its bounds will require monotonicity of selection). The quantities $(\alpha,\beta)$ may be of interest in their own right, but they also play a crucial role in constructing bounds and sensitivity analyses for the SCATE $\psi$, since they are the unidentified parts of $\psi$.  

\section{Bounds, Sensitivity Analysis, \& Estimation}
\label{sec:methods}

Next, we derive bounds on the SCATE $\psi$, and we use these bounds to construct optimally informative and interpretable sensitivity analysis parameters; we then detail flexible yet efficient estimation methods.

\subsection{Bounds}

Although the SCATE $\psi$ is not point-identified, we show here that it is possible to construct potentially informative bounds. Bounds are often used to derive inferences under weaker conditions than those needed for point identification \autocite{GriMea2008,schwartz2012sensitivity,zhang2008evaluating}. Importantly, our analysis exploits covariate information to sharpen bounds as much as possible; then we use these worst-case bounds to aid a more informative sensitivity analysis. 

For notational simplicity, we denote
\begin{align*}
\pi_z(\bx) &= \Pb(Z=z \mid \bX=\bx), \\
\theta_z(a \mid \bx) &= \Pb(A=a \mid \bX=\bx,Z=z), \\
\lambda_z(\bx) &= \E(S \mid \bX=\bx,Z=z), \\
\mu_z(\bx) &= \E(YS \mid \bX=\bx,Z=z) 
\end{align*}
as important nuisance functions that will be used throughout. Note that since $A \in \{0,1,\cdot\}$, we have $1-\theta_z(0 \mid \bx)=\theta_z(1 \mid \bx) + \theta_z(\cdot \mid \bx)$ where $A=\cdot$ if $S=0$. Also, as pointed out earlier, we focus on the case where $Y$ is binary, although our results equally apply to arbitrary outcomes as long as they are bounded.

In the next theorem we provide bounds $(\alpha_\ell,\alpha_u)$ on the survivor-complier proportion $\alpha$ and bounds $(\beta_\ell,\beta_u)$ on the survivor intention-to-treat effect $\beta$. We will use these bounds to construct bounds and a sensitivity analysis for the SCATE $\psi$.

\begin{theorem}
Under Assumptions 1--6, the survivor-complier proportion is bounded as 
$$ \E[ \{ \theta_0(0 \mid \bX) - \theta_1(0 \mid \bX) \}_+] \leq \alpha \leq \E\{ \theta_1(1 \mid \bX) - \theta_0(1 \mid \bX) \} , $$
and the survivor intention-to-treat effect is bounded as 
$$ \frac{\E[ \{ \mu_1(\bX) + \lambda_0(\bX) -  \lambda_1(\bX) \}_+ - \mu_0(\bX)] }{ \E\{ \lambda_0(\bX)\} } \leq \beta \leq \frac{ \E[ \{ \mu_1(\bX)  \wedge  \lambda_0(\bX) \} - \mu_0(\bX) ] } { \E\{ \lambda_0(\bX)\} } . $$
\end{theorem}

The intuition behind the bound on $\alpha$ is as follows. Under Assumptions 1--6, the quantity $\theta_1(1 \mid \bX)-\theta_0(1 \mid \bX) = \Pb(A^{(z=1)}=1 \mid \bX) - \Pb(A^{(z=0)}=1 \mid \bX)$ is the difference in proportions of units exposed to the treatment (less than a 4 hour wait) when encouraged versus not among units with covariates $\bX$. This equals the proportion with covariates $\bX$ who either are survivor-compliers or are selection-compliers ($S^{(z=1)}>S^{(z=0)}$) that take treatment when encouraged, which is the upper bound on the proportion of survivor-compliers. Averaging the covariate-specific upper bound gives an upper bound on the marginal proportion of survivor-compliers. Similar logic shows that $\theta_0(0 \mid \bX)-\theta_1(0 \mid \bX) = \Pb(A^{(z=1)} \neq 0 \mid \bX) - \Pb(A^{(z=0)} \neq 0 \mid \bX)$ equals the covariate-specific proportion of survivor-compliers minus the proportion of selection-compliers that take control when encouraged, which is clearly a lower bound. Note that when there is no selection (i.e., $S=1$ with probability one), the lower and upper bounds collapse to $\E(A \mid \bX,Z=1)-\E(A \mid \bX,Z=0)$, which is equivalent to the denominator of the usual Wald estimator 

The bounds on $\beta$ are a covariate-adjusted version of those previously given by \textcite{hudgens2003analysis, zhang2003estimation}, since $\beta$ is the survivor average causal effect of encouragement by the IV. As with $\alpha$, when there is no selection, the bounds on $\beta$ collapse to the numerator of the usual Wald estimand given by $\E(Y \mid \bX,Z=1)-\E(Y \mid \bX,Z=0)$, since in that case we have $\lambda_z(\bX)=1$. In fact the lower bound $\beta_\ell$ can collapse to a scaled version of the Wald numerator even in the presence of selection, whenever there is no differential selection for those encouraged or not, i.e., when $\Pb\{\lambda_1(\bX)=\lambda_0(\bX)\}=1$ or equivalently $S \ind Z \mid \bX$.

Using the bounds for $\alpha$ and  $\beta$, we can construct bounds on the SCATE $\psi$, since the other term in the numerator is identified. These bounds are given in the next theorem.

\begin{theorem}
Let $(\alpha_\ell,\alpha_u)$ and $(\beta_\ell,\beta_u)$ denote the lower and upper bounds for $\alpha$ and $\beta$, respectively, given in Theorem 1. Under Assumptions 1--6, the SCATE is bounded as
$$ \frac{ \beta_\ell \ \E\{\lambda_0(\bX)\} }{ (\alpha_u-\alpha_\ell) \one(\beta_\ell > 0) + \alpha_\ell } \leq \psi \leq \frac{ \beta_u \ \E\{\lambda_0(\bX)\} }{ (\alpha_\ell -\alpha_u) \one(\beta_u > 0) + \alpha_u } . $$
\end{theorem}

 As with the bounds on $\alpha$ and $\beta$, the bounds $(\psi_\ell,\psi_u)$ on the SCATE collapse to the usual Wald estimator when there is no selection. However, even in the presence of selection, the bounds operate in a fashion analogous to the Wald estimator. In each case, we have a ratio estimator based on the bounds for $\alpha$ and $\beta$. The numerators are bounds on $\beta$, which is the analog of the intention-to-treat effect, and the denominators are bounds on $\alpha/\Pb(S^{(z=1)}=S^{(z=0)}=1)$, which is the proportion of compliers among the survivors. The correction factor $(\alpha_\ell -\alpha_u)$ picks out the appropriate bound on $\alpha$ depending on the sign for the bound on $\beta$. When $S \ind Z \mid \bX$, the width of the bounds is completely determined by the non-identification of the proportion of survivor-compliers $\alpha$.
 
 While the width of the bounds depends on the use of covariates, not all covariates are equally useful in narrowing the bounds. Covariates that are strong predictors of selection $S$ are especially important. In the extreme case where $S=g(\bX)$ for some $g$ (i.e., covariates can perfectly predict selection), then the bounds collapse to a single point. A proof of this fact is given in the Appendix. In intermediate cases, where the covariates for $S$ are good but not perfect predictors, the bounds can still be much tighter than those not incorporating covariates; this is illustrated via simulations in Section 4. These results show how covariates can be critical to improving the bounds. They should also help guide data collection in that researchers should focus on covariates that are strong predictors of selection.
 
 Apart from the cases mentioned above where covariates are perfectly predictive, the only way to achieve point identification of the SCATE $\psi$ is to add stronger assumptions. For example, if one assumed that the third stratum in Table 1 was empty, i.e., that $\Pb(S^{(z=1)}>S^{(z=0)},A^{(z=1)}=1 \mid \bX)=0$, then $\alpha$ would be identified with its upper bound $\alpha=\alpha_u$. Conversely, if one assumed the second stratum in Table 1 was empty, i.e., that $\Pb(S^{(z=1)}>S^{(z=0)},A^{(z=1)}=0 \mid \bX)=0$, then $\alpha$ would be identified with its lower bound $\alpha=\alpha_\ell$. Since $\beta$ is a survivor average effect of IV encouragement, it could be identified by adding assumptions considered by previous authors, such as parametric modeling assumptions \autocite{zhang2009likelihood}, structural assumptions about pre-treatment covariates \autocite{ding2011identifiability}, or no unmeasured confounding assumptions given baseline or post-IV covariates \autocite{tchetgen2014identification}. In our view, additional identifying assumptions of this type are likely implausible in the ICU application.

\subsection{Sensitivity Analysis}

In some cases, the bounds given in Theorem 2 will be sufficiently narrow so as to be informative and practically meaningful. In other cases, however, they may be wide and thus unsatisfactory in that they only give a worst-case snapshot of possible SCATE values. In the latter case, we propose that it is more informative to display all possible values of $\psi$ as a function of interpretable sensitivity parameters, not just the two most extreme worst-case possibilities. This is in the same spirit as \textcite{IchMeaNan2008}, for example, since the goal is to use the data to inform more specific inferences about $\psi$. 

In our setting, there are two natural sensitivity parameters: the survivor-complier proportion $\alpha$ and the survivor intention-to-treat effect $\beta$. We can vary these sensitivity parameters by interpolating between their upper and lower bounds given in Theorem 1. Therefore we consider the parameter
$$ \psi^*(\bdelta) = \left\{ \frac{\delta_2 \beta_u + (1-\delta_2) \beta_\ell }{ \delta_1 \alpha_u + (1-\delta_1) \alpha_\ell } \right\} \E\{\lambda_0(\bX)\} $$
for $(\delta_1,\delta_2) = \bdelta \in [0,1]^2$. Note that larger values of $\delta_1$ correspond to larger values of $\alpha$, and larger values of $\delta_2$ correspond to larger values of $\beta$, with $\delta_1=1$ recovering $\alpha_u$ and $\delta_1=0$ recovering $\alpha_\ell$, and similarly for $\beta$. And of course we have $\psi_\ell \leq \psi^*(\bdelta) \leq \psi_u$ for any $\bdelta$, where $(\psi_\ell, \psi_u)$ are the bounds on the SCATE given in Theorem 2. However, consideration of $\psi^*(\bdelta)$ allows for more nuance than just $(\psi_\ell,\psi_u)$, since $\psi^*(\bdelta)$ allows us to incorporate prior knowledge about the sizes of the survivor-complier proportion $\alpha$ or the intention-to-treat effect $\beta$. 

Although we index $\psi^*(\bdelta)$ with the parameter $\bdelta$ to simplify notation and theoretical analysis, in practice (e.g., when plotting results) it will be more useful to index estimates by the implied values of $(\alpha, \beta)$, i.e., by $\alpha^* = \delta_1 \alpha_u + (1-\delta_1) \alpha_\ell$ and $\beta^*=\delta_2 \beta_u + (1-\delta_2) \beta_\ell$.  More specifically, we index $\psi^*(\bdelta)$ by $\bdelta$ to reflect the fact that the implied values $(\alpha^*,\beta^*)$ are not arbitrarily chosen, but instead depend on the bounds $(\alpha_\ell,\alpha_u)$ and $(\beta_\ell,\beta_u)$ implied by $\Pb$. This will be important when accounting for uncertainty in estimating $\psi^*(\bdelta)$ due to the fact that we have estimated the bounds for $(\alpha,\beta)$.

Note that what we propose is different from a standard sensitivity analysis, which typically relies on fixed arbitrary values of sensitivity parameters. In contrast, our sensitivity parameters are informed by the data. In other words, we do not need to consider the entire range of values $\alpha^* \in (0,1]$ and $\beta \in [-1,1]$, instead we can consider the narrower range of values consistent with the observed data. This will typically lead to a more meaningful and interpretable sensitivity analysis.

\subsection{Estimation \& Inference}

Next, we construct efficient nonparametric estimators for the parameter $\psi^*(\bdelta)$ and detail their asymptotic properties, showing in particular how to construct confidence intervals. Our approach to estimation makes use of semi- and non-parametric efficiency theory and influence functions \autocite{bickel1993efficient,van2003unified,tsiatis2006semiparametric}. Influence functions are essential for a number of reasons. First, the influence function with the smallest variance (called the efficient influence function) determines the efficiency bound for estimating a parameter, thus giving a benchmark against which estimators can be measured and indicating the difficulty of the problem. Second, influence functions can be used to construct estimators with very favorable properties, such as double robustness or general second-order bias. For this reason, influence function-based estimators can attain fast parametric rates of convergence, even while allowing flexible nonparametric estimation of complex high-dimensional nuisance functions. 

First, we outline additional notation. We let $\Pn$ denote the empirical measure so that for example sample averages can be written as $n^{-1} \sum_i f(\bO_i)= \Pn(f) = \Pn\{f(\bX)\}$. Further, for any arbitrary variable $T \subset \bO$, we let
$$ \phi_z(T) = \frac{\one(Z=z)}{\pi_z(\bX)}\Big\{ T - \E(T \mid \bX, Z=z) \Big\} + \E(T \mid \bX, Z=z) $$
denote the uncentered component of the efficient influence function for the parameter $\E\{\phi_z(T)\} = \E\{\E(T \mid \bX, Z=z)\}$. Note that, in addition to $z$ and $T$, the quantity $\phi_z(T)$ also depends on the variables $(\bX,Z)$ as well as the nuisance functions $\pi_z(\bx)$ and $\E(T \mid \bX=\bx,Z=z)$; for now we keep this dependence implicit for notational simplicity. Similarly we let 
$$ \hat\phi_z(T) = \frac{\one(Z=z)}{\hat\pi_z(\bX)}\Big\{ T - \hat\E(T \mid \bX, Z=z) \Big\} + \hat\E(T \mid \bX, Z=z) $$
denote the estimated version of $\phi_z(T)$, based on nuisance estimators $\hat\pi_z(\bX)$ and $\hat\E(T \mid \bX, Z=z)$. Note that we are largely agnostic about which particular methods are used to construct these nuisance estimators, leaving this up to the analyst. Our theory allows for either parametric or nonparametric estimators, with asymptotic normality and root-$n$ rates requiring certain convergence rate conditions, as detailed in Theorem 3. The functions $\phi_z$ and $\hat\phi_z$ will play an important role in constructing efficient estimators, analyzing their asymptotic behavior, and computing confidence intervals.

Before discussing our proposed estimators, we need to point out two regularity conditions required for such estimators to be well-behaved. These conditions are required since the bounds $(\alpha_\ell, \beta_\ell, \beta_u)$ from the previous subsection are nonsmooth, i.e., $\alpha_\ell$ and $\beta_\ell$ involve maxima, and $\beta_u$ involves a minimum. In particular, without such conditions the parameter $\psi^*(\bdelta)$ would not admit an influence function, and regular estimators would not exist. Informally, regular estimators are nicely behaved estimators whose asymptotic distribution is invariant to small shifts in the data-generating process; for more details see \textcite{van2000asymptotic, tsiatis2006semiparametric, hirano2012impossibility}. The two regularity conditions we need are as follows.

\begin{condition}
$\Pb\{ \theta_1(0 \mid \bX) = \theta_0(0 \mid \bX) \} = 0$.
\end{condition}
\begin{condition}
$\Pb[ \mu_1(\bX) \in\{ \lambda_0(\bX), \lambda_0(\bX)-\lambda_1(\bX)\} ] = 0$.
\end{condition}
Conditions 1--2 are similar in spirit to so-called non-exceptional law conditions in optimal treatment regime estimation \autocite{robins2004optimal,van2014targeted}. However, our conditions may be more plausible since the quantities they restrict are not conditional effects of the treatment $A$, which might reasonably be zero in practice.

Since under Assumptions 1--6 we have $\theta_z(0 \mid \bX)=\Pb(A^{(z)}=0 \mid \bX)$, Condition 1 says that there are almost no covariate strata in which the proportion of survivor-compliers equals the proportion of units in the second stratum in Table 1 (i.e., compliers with respect to selection, who take control when selected). In the ICU application this means that, for almost all covariate strata, the proportion of patients who would always be admitted to the ICU regardless of bed availability, but whose wait time would depend on availability (i.e., survivor-compliers), cannot equal the proportion of patients who would only be admitted to the ICU with more beds available, and who would wait more than 4 hours if admitted. This appears to be a plausible assumption in the ICU application; it would be surprising if these particular strata happened to be exactly the same size. 

Condition 2 is similar, though perhaps slightly less intuitive. Since under Assumptions 1--6 we have $\mu_z(\bX)=\E(Y^{(z)}=1 \mid \bX, S^{(z)}=1)\lambda_z(\bX)$ and $\lambda_z(\bX)=\Pb(S^{(z)}=1 \mid \bX)$, Condition 2 says that there are almost no covariate strata in which the mean outcome under treatment for those selected under encouragement from the IV equals either (a) the ratio $\Pb(S^{(z=0)}=1 \mid \bX)/\Pb(S^{(z=1)}=1 \mid \bX)$ of those who would be selected under no encouragement versus encouragement, or (b) this ratio minus one. Note that we must have $\lambda_1(\bX) \geq \lambda_0(\bX)$ under monotonicity, so the only way the latter part of Condition 2 could be violated is if for some non-negligible strata there are only always- and never-takers with respect to selection, and in these exact strata the outcome for the always-takers happens to be always zero. As with Condition 1, this is a somewhat contrived scenario that we would generally not expect to encounter in practice.

Now we are ready to present our proposed estimator for $\psi^*(\bdelta)$, which is based on estimates of the bounds on $\alpha$ and $\beta$. Specifically we let
\begin{align*}
\hat\alpha_\ell &= \Pn\left( \one\Big\{ \hat\theta_0(0 \mid \bX)> \hat\theta_1(0 \mid \bX) \Big\} \Big[ \hat\phi_1\{\one(A \neq 0)\} - \hat\phi_0\{\one(A \neq 0)\} \Big] \right) \\
\hat\alpha_u &= \Pn\left[ \hat\phi_1\Big\{\one(A =1)\Big\} - \hat\phi_0\Big\{\one(A = 1)\Big\} \right] \\
\hat\beta_\ell &= \Pn\left( \one\left\{ \frac{\hat\mu_1(\bX)}{\hat\lambda_1(\bX)-\hat\lambda_0(\bX)} > 1 \right\} \Big[ \hat\phi_1\{S(Y-1)\} - \hat\phi_0(S) \Big] - \hat\phi_0(SY) \right) \Big/ \Pn\Big\{ \hat\phi_0(S) \Big\}  \\
\hat\beta_u &= \Pn\left[ \one\left\{  \frac{\hat\mu_1(\bX)}{\hat\lambda_0(\bX)} > 1 \right\} \Big\{ \hat\phi_0(S) - \hat\phi_1(SY) \Big\} + \hat\phi_1(SY) - \hat\phi_0(SY) \right] \Big/ \Pn\Big\{ \hat\phi_0(S) \Big\} 
\end{align*}
denote efficient influence-function-based estimators of the bounds on $(\alpha,\beta)$, and then the corresponding estimator of $\psi^*(\bdelta)$ is given by
$$ \widehat\psi^*(\bdelta) = \left\{ \frac{\delta_2 \hat\beta_u + (1-\delta_2) \hat\beta_\ell }{ \delta_1 \hat\alpha_u + (1-\delta_1) \hat\alpha_\ell } \right\} \Pn\Big\{ \hat\phi_0(S) \Big\} . $$
In the Appendix we show that this estimator solves an estimating equation based on the efficient influence function for $\psi^*(\bdelta)$; this is a standard technique for constructing efficient estimators. 

The next theorem gives conditions under which our proposed estimator is asymptotically normal and efficient in a nonparametric model, which puts no constraints on the probability distribution $\Pb$ (other than, at most, nonparametric smoothness, sparsity, or other structural constraints). To ease notation we let $\theta_{za}=\theta_z(a \mid \bx)$ and generally suppress the dependence of the nuisance functions on $\bX$. We also let $\|f\|^2 = \int f(\bx)^2 \ d\Pb(\bx)$ denote the squared $L_2(\Pb)$ norm.

\begin{theorem} \label{thm:asymp}
Along with Assumptions 1--6 and Conditions 1--2, assume
\begin{enumerate}
\item $\{\widehat\psi^*(\bdelta), \hat{\boldsymbol\eta}\} \inprob \{\psi^*(\bdelta), \boldsymbol\eta\}$ for true value $\boldsymbol\eta = (\pi_z, \theta_{za}, \lambda_z, \mu_z)$.
\item The influence functions $\boldsymbol{\varphi}({\boldsymbol\eta}) = (\varphi_\ell^{(\alpha)}, \varphi_u^{(\alpha)}, \varphi_\ell^{(\beta)}, \varphi_u^{(\beta)})$ (as defined in the Appendix) and their estimates $\boldsymbol{\varphi}(\hat{\boldsymbol\eta})$ fall in a Donsker class with probability one as $n \rightarrow \infty$, and are estimated consistently in the sense that $\| \boldsymbol{\varphi}({\boldsymbol\eta}) - \boldsymbol{\varphi}(\hat{\boldsymbol\eta}) \|^2=o_\Pb(1)$.
\item $\| \hat\pi_1 - \pi_1 \| \Big( \max_{z,a} \| \hat\theta_{za} - \theta_{za} \| + \max_z \| \hat\lambda_z - \lambda_z \| + \max_z \| \hat\mu_z - \mu_z \| \Big) = o_\Pb(1/\sqrt{n})$.
\item $\| \hat\eta_j - \eta_j \| \sqrt{ \Pb(|\eta_j| < |\hat\eta_j - \eta_j |) } = o_\Pb(1/\sqrt{n})$ for $\eta_j \in \{ \theta_{10}-\theta_{00}, \mu_1-(\lambda_1-\lambda_0), \mu_1-\lambda_0 \}$.
\end{enumerate}
Then
$$ \sqrt{n}\Big\{ \widehat\psi^*(\bdelta) - \psi^*(\bdelta) \Big\} \indist N\left( 0, \var\left[ \frac{\varphi^{(\beta)}(\delta_2) - \psi^*(\bdelta) \varphi^{(\alpha)}(\delta_1)}{\E\{\varphi^{(\alpha)}(\delta_1)\} } \right] \right) $$
where
$$ \varphi^{(\alpha)}(\delta_1) = \delta_1 \varphi_u^{(\alpha)} + (1-\delta_1) \varphi_\ell^{(\alpha)}  \ , \ \varphi^{(\beta)}(\delta_2) = \delta_2 \varphi_u^{(\beta)} + (1-\delta_2) \varphi_\ell^{(\beta)}$$
for $(\varphi_\ell^{(\alpha)}, \varphi_u^{(\alpha)}, \varphi_\ell^{(\beta)}, \varphi_u^{(\beta)})$ the efficient influence functions as defined in the Appendix.
\end{theorem}

Theorem \ref{thm:asymp} is critically important. It shows that our proposed estimator is $\sqrt{n}$ consistent and asymptotically normal even if the nuisance functions $\boldsymbol\eta$ are estimated at slower than parametric $\sqrt{n}$ rates, as long as they attain the $n^{1/4}$-type rates described in Conditions 3--4 of Theorem \ref{thm:asymp}. Importantly, such $n^{1/4}$ rates are attainable under sparsity, smoothness, or other nonparametric structural constraints, whereas the $\sqrt{n}$ rates required by maximum-likelihood-type estimators necessitate the use of correct parametric models (which can be difficult if not impossible to specify even in the presence of a few continuous covariates). Further, 95\% confidence intervals can be easily constructed as $\widehat\psi^*(\bdelta) \pm 1.96 \hat\sigma(\bdelta)/\sqrt{n}$, using empirical estimates of the asymptotic variance $\hat\sigma^2=\Pn\{( \hat\varphi^{(\beta)} - \widehat\psi^* \hat\varphi^{(\alpha)})^2\} / \{\Pn( \hat\varphi^{(\alpha)})\}^2$ (the dependence on $\bdelta$ is suppressed for notational simplicity). This does not require any extra estimation or model-fitting since the estimated influence functions $\hat\varphi$ are already computed to construct the estimator $\widehat\psi^*$. Condition 2 of Theorem \ref{thm:asymp} puts a mild restriction on the complexity of the nuisance estimators we use, however this can be completely avoided by sample splitting \autocite{zheng2010asymptotic}. In summary, Theorem \ref{thm:asymp} allows for efficient $\sqrt{n}$-rate estimation and valid inference (e.g., confidence intervals), even though our analysis is  nonparametric and can incorporate complex machine learning methods.

\subsection{Bounding the Effect of ICU Wait Times}

We estimate the bounds for SCATE $\psi$ both with and without covariates to understand how covariate information might narrow the bounds. Without using covariates, the upper and lower bounds on the survivor-complier proportion $\alpha$ are 0 and 0.08, respectively.  Thus the percentage of patients that would wait less than four hours to gain access to the ICU due to bed availability is between 0 and 8\%. The upper and lower bounds for $\beta$ are -0.19 and 0.06.  As such, bed availability among those always admitted to the ICU in less than four hours may have reduced mortality by 19\% or increased mortality by 6\%.  The corresponding bounds for the SCATE $\psi$ are -1 and 1 and thus are uninformative.

We next estimate these bounds exploiting covariate information to weaken identifying assumptions and gain more efficiency. We model the covariate relationships (i.e., nuisance functions $\boldsymbol\eta$) flexibly using Super Learner \autocite{van2011targeted}, combining logistic regression, generalized additive models, lasso, and multivariate adaptive regression splines.  With covariates, the bounds on $\alpha$ are 0.01 and 0.07, and the bounds on $\beta$ are -0.14 and 0.05. The bounds for SCATE $\psi$ remain -1 and 1.  Thus although incorporating covariate information sharpens the bounds on $\alpha$ and $\beta$ considerably (reducing length by 23\% for $\alpha$ and 22\% for $\beta$), it does not sharpen the bounds for the ratio $\psi$. Next, we conduct a sensitivity analysis where we estimate $\psi$ using the estimated bounds on $\alpha$ and $\beta$ as sensitivity parameters. This analysis creates a grid of possible $\psi$ estimates, and for each of these $\psi$ estimates, we construct point-wise confidence intervals.  The results are contained in Figure~\ref{fig:sens.plots}. These results are substantially more informative than the bounds on $\psi$ alone.  We observe that the preponderance of estimates for $\psi$ indicate that shorter waiting times for the ICU reduces mortality.  The sensitivity plot does not rule out the possibility that shorter waiting times increase mortality. However, the bounds  clearly favor the possibility that prompt ICU care reduces mortality. Moreover, analyses of these same data using an IV analysis found that ICU care generally lowers mortality rates \autocite{keele2016stronger}.  While that estimate is based on a different sub-population, finding that ICU care is generally beneficial in conjunction with our analysis provides further evidence that prompt ICU care is beneficial for patients. 

\begin{figure}
\spacingset{1}
    \centering
\includegraphics[width=.7\textwidth]{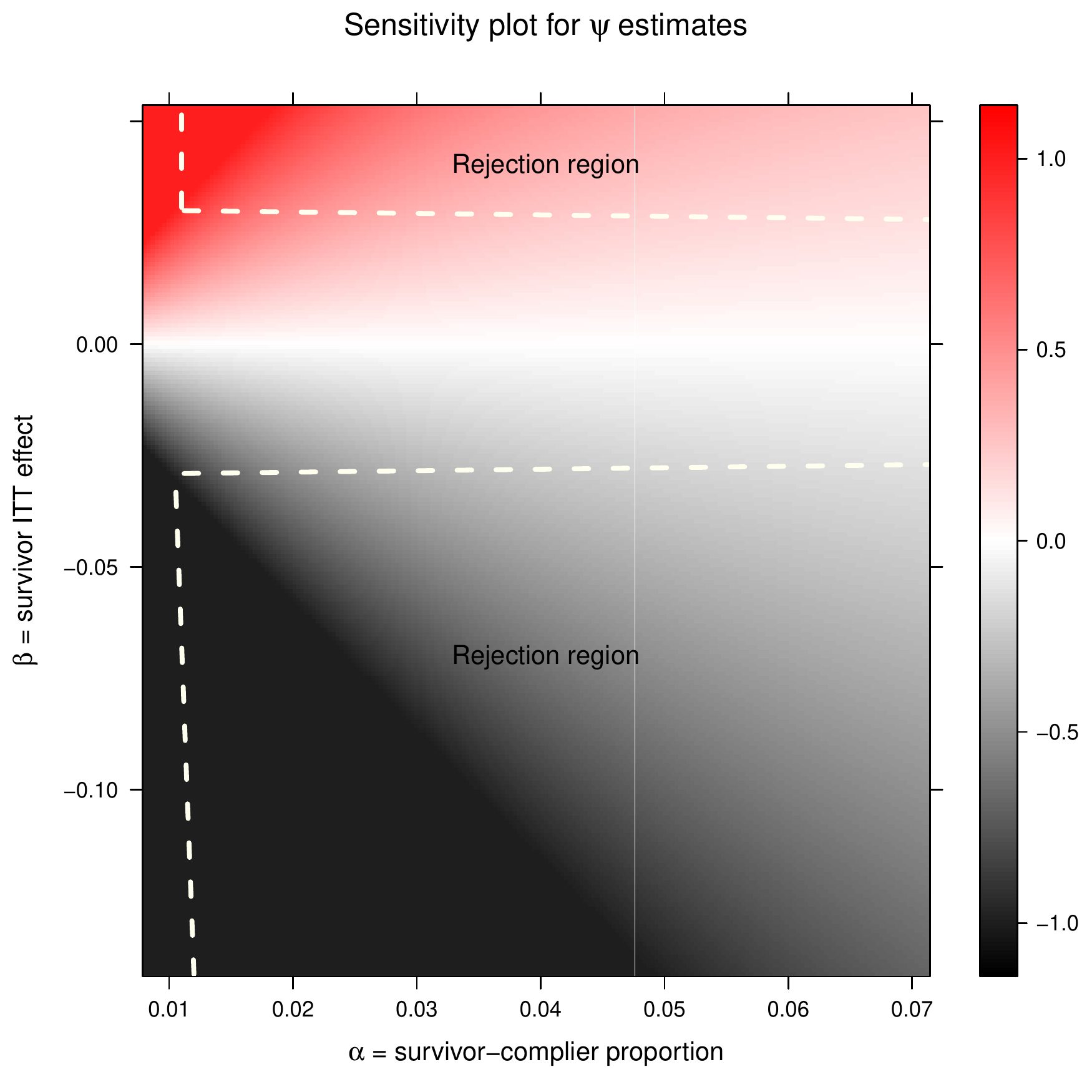}
\caption{Sensitivity plot for the effect of ICU waiting time on mortality. Points in the rejection region do not include zero in a 95\% point-wise confidence interval.}
\label{fig:sens.plots}
\end{figure}

\section{Simulation}
\label{sec:sim}

Finally, we conduct a simulation to better understand under what conditions the bounds will be informative. Here, we simulate data from the true causal model and calculate the bounds while varying the magnitude of the true values for $\alpha$ and $\psi$. The sample size for each simulation is 1000, and the simulations were repeated 1000 times for each combination of the true $\alpha$ and $\psi$ parameters. In all scenarios, the instrumental variable, $Z$, was generated as a Bernoulli random variable with the probability of success set to 0.5.  We simulated principal strata membership probabilistically, with the probability of being in the complier-survivor strata set to $\alpha$, and the probability of being in the other five principal strata as $(1- \alpha)/5$. Within the complier-survivor stratum, $Y$ is a draw from a binomial distribution with probability $(0.5 - \psi)/2$, and is a draw from a binomial distribution with probability 0.5 for the over five principal strata. We  calculated the bounds for $\psi$ using values of 0 to 1 for $\alpha$ and -1 to 1 for $\psi$ in increments of 0.01. For each simulation, we recorded the length of the bounds and whether the sign for the lower and upper bound agree. These two quantities capture how much information is in the bounds as a function of whether the instrument encourages $A=1$ and the magnitude of the true causal effect. When $A$ is not highly affected by the instrument and when the true causal effect is small the bounds are unlikely to be informative. However, of greater interest is the point at which the bounds can at least be informative about the sign of the true causal effect.

In the simulation, we also studied how covariates can narrow the bounds. In the first set of simulations, we did not condition on covariates. In two additional simulations, we generated $S$ as a function of two Normally distributed covariates. In one simulation, the covariates were more weakly informative for $S$ as both were correlated with $S$ at approximately 0.50; in the second set, the two covariates were more informative as both were correlated with $S$ at approximately 0.85.

We summarize the results from the simulations in Figure~\ref{fig:sim.plots}.  First, we examine the length of the bounds. In this plot, we observe that the overall length of the bounds is primarily a function of the value of $\alpha$. That is, the bounds will be at their maximum length of 2 when $\alpha$ is approximately 0.40 or less regardless of the size of $\psi$. Next, we plot when the bounds are informative as to the sign of the true causal effect. Here, the results depend on the magnitude of both $\alpha$ and $\psi$. We observe that when the $\psi$ is larger and $\alpha$ is less than 0.40, the bounds may be informative about the sign of the causal effect, but given the information in the two other panels the bounds are likely to be quite wide.  However, once $\alpha$ is larger than 0.40 the information about the sign of the causal effect depends directly on the magnitude of $\psi$. Moreover, for larger values of $\alpha$ the bounds can be highly informative. For example, if $\psi$ is 0.10, and $\alpha$ is 0.65, the bounds are 0.04 and 0.26. If $\alpha$ increases to 0.75 and $\psi$ remains 0.10, the bounds narrow to 0.04 and 0.18.  If $\alpha$ remains 0.65, but $\psi$ increases to 0.20, the bounds actually widen to 0.06 and 0.42. 

Figure~\ref{fig:sim.plots} also shows results from the simulations where covariates were used to construct the bounds. First, we observe that even with the covariates that are more weakly predictive of $S$, the width of the bounds decreases and the region in which the sign of the bounds is informative is much larger. Across all the values for $\alpha$ and $\psi$ in the simulations, the average length of the bounds without covariates was 0.91. With the weakly predictive covariates,

\begin{figure}[H]
\spacingset{1}
\centering
\begin{subfigure}{.45\textwidth}
  \centering
  \includegraphics[width=.9\linewidth]{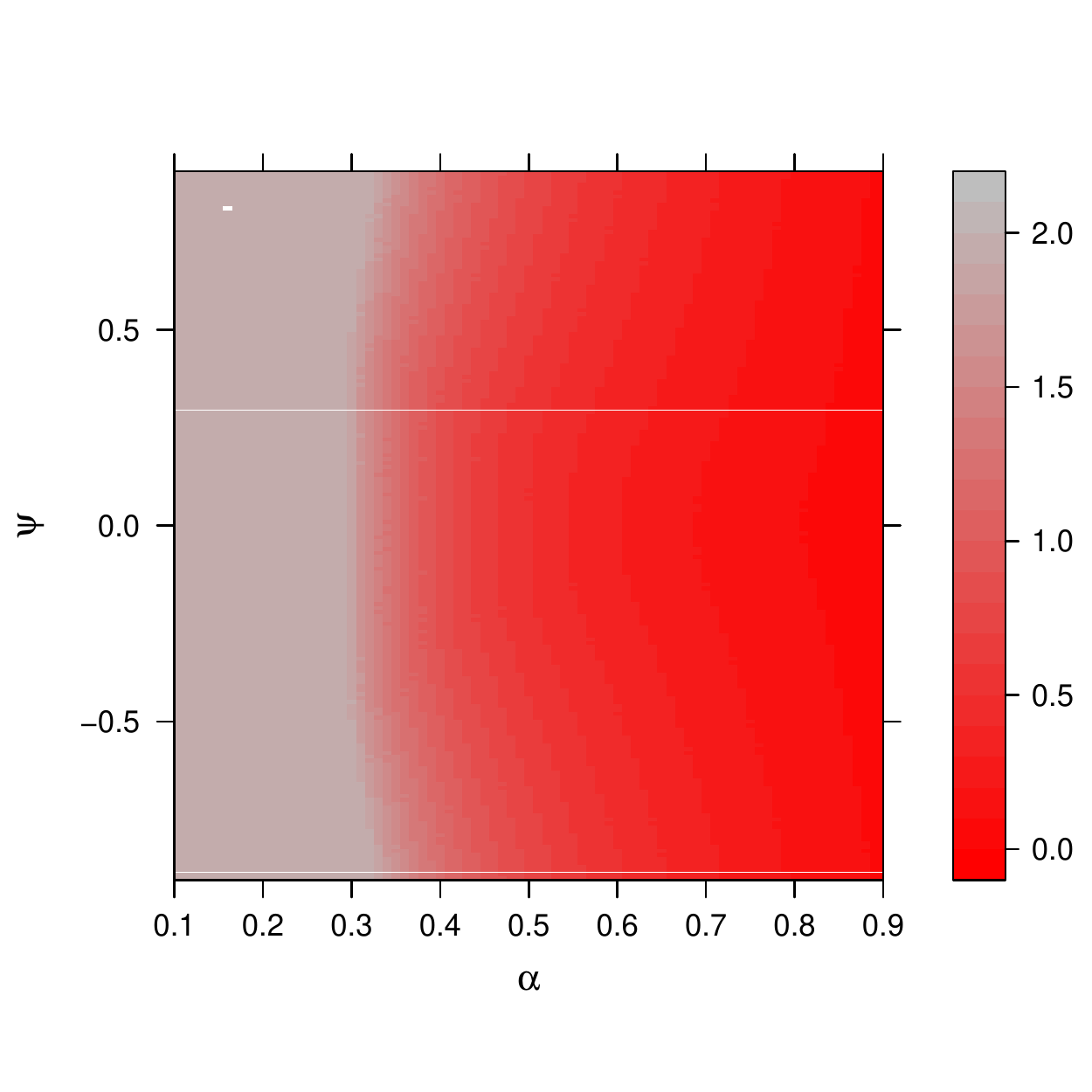}
  \vspace{-5mm}
  \caption{\footnotesize{Length of Bounds - No Covariates}}
  \label{fig:sub1}
\end{subfigure}
\begin{subfigure}{.45\textwidth}
  \centering
  \includegraphics[width=.9\linewidth]{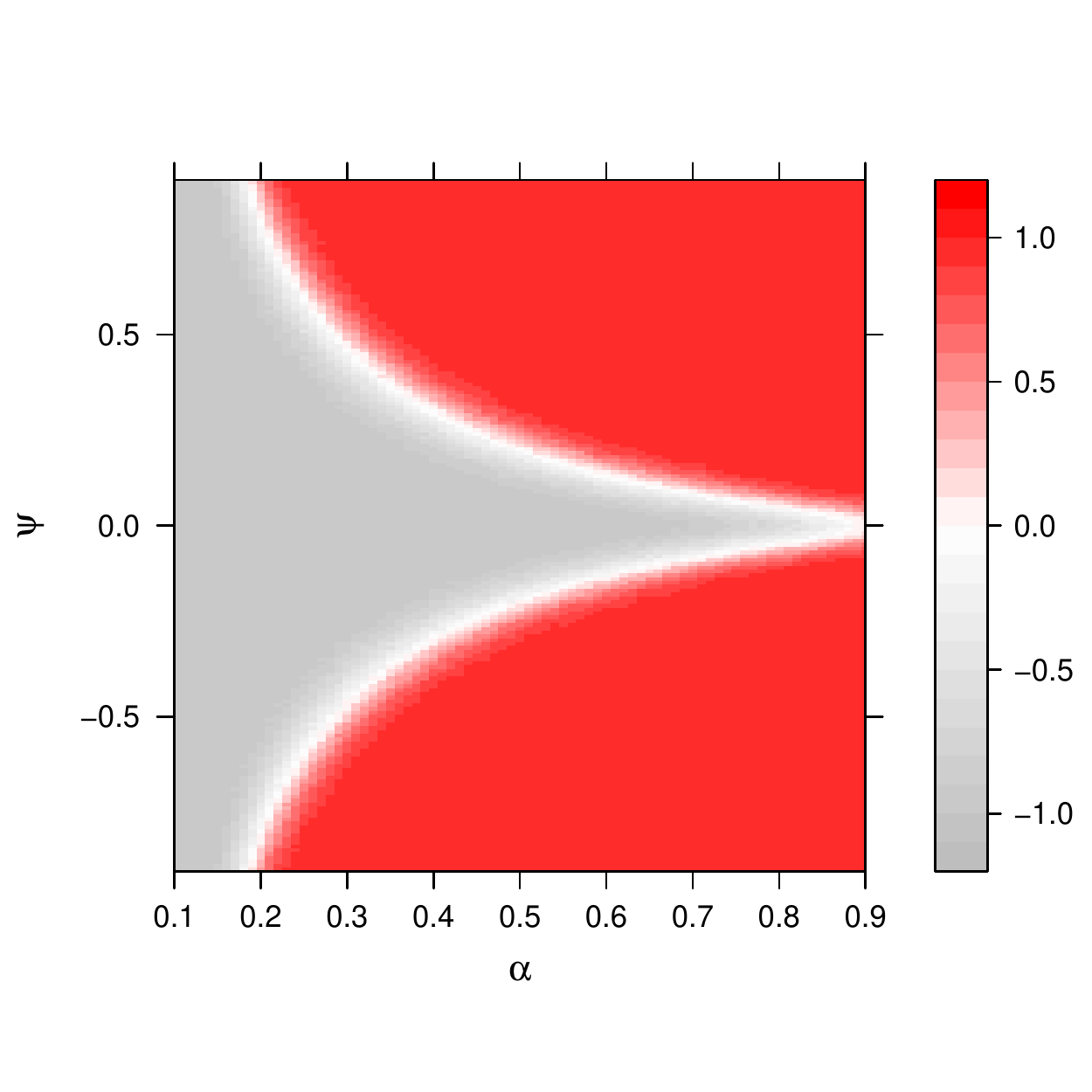}
  \vspace{-5mm}
  \caption{\footnotesize{Informative Sign - No Covariates}}
  \label{fig:sub2}
\end{subfigure}
\begin{subfigure}{.45\textwidth}
  \centering
  \includegraphics[width=.9\linewidth]{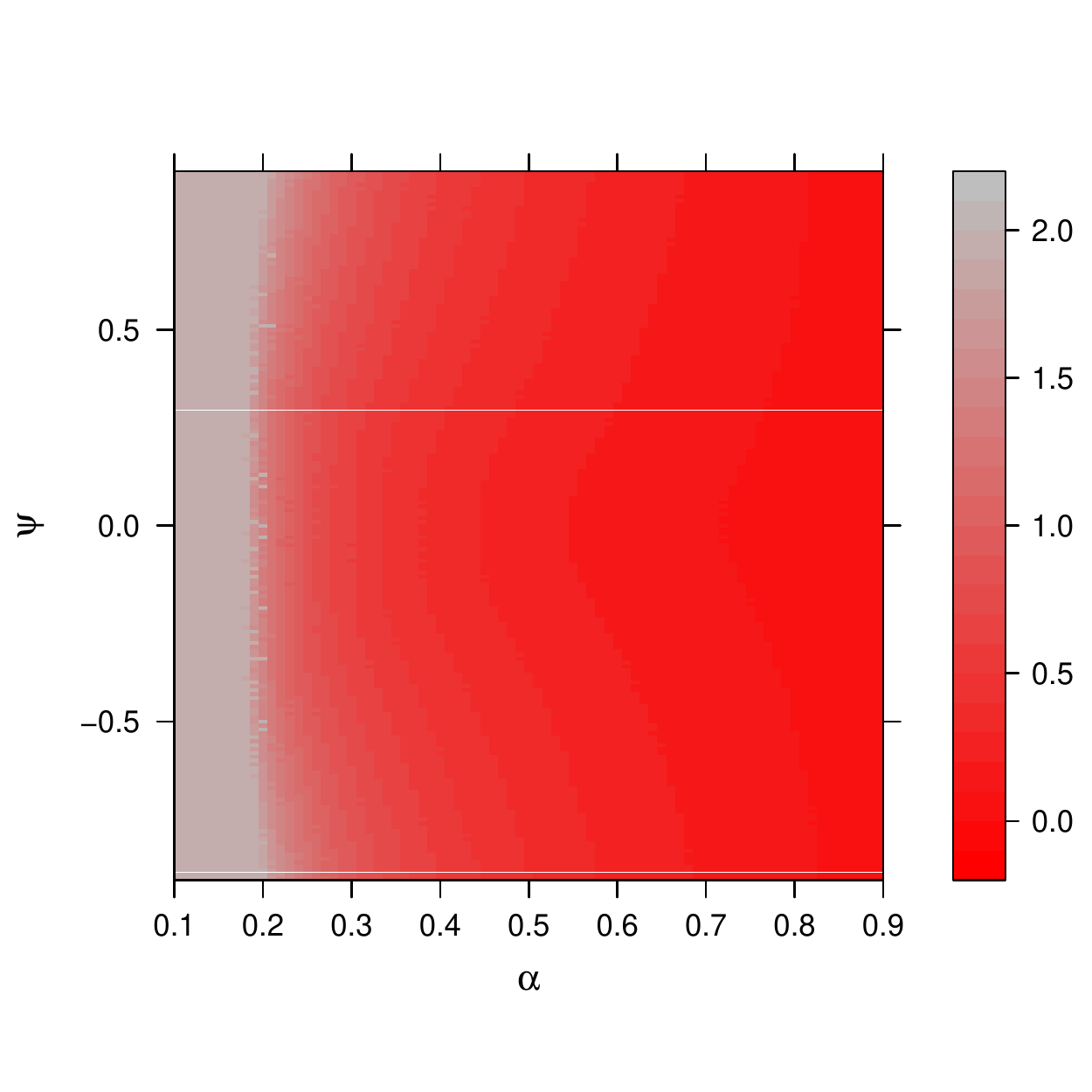}
  \vspace{-5mm}
  \caption{\footnotesize{Length of Bounds - Weak Covariates}}
  \label{fig:sub3}
\end{subfigure}
\begin{subfigure}{.45\textwidth}
  \centering
  \includegraphics[width=.9\linewidth]{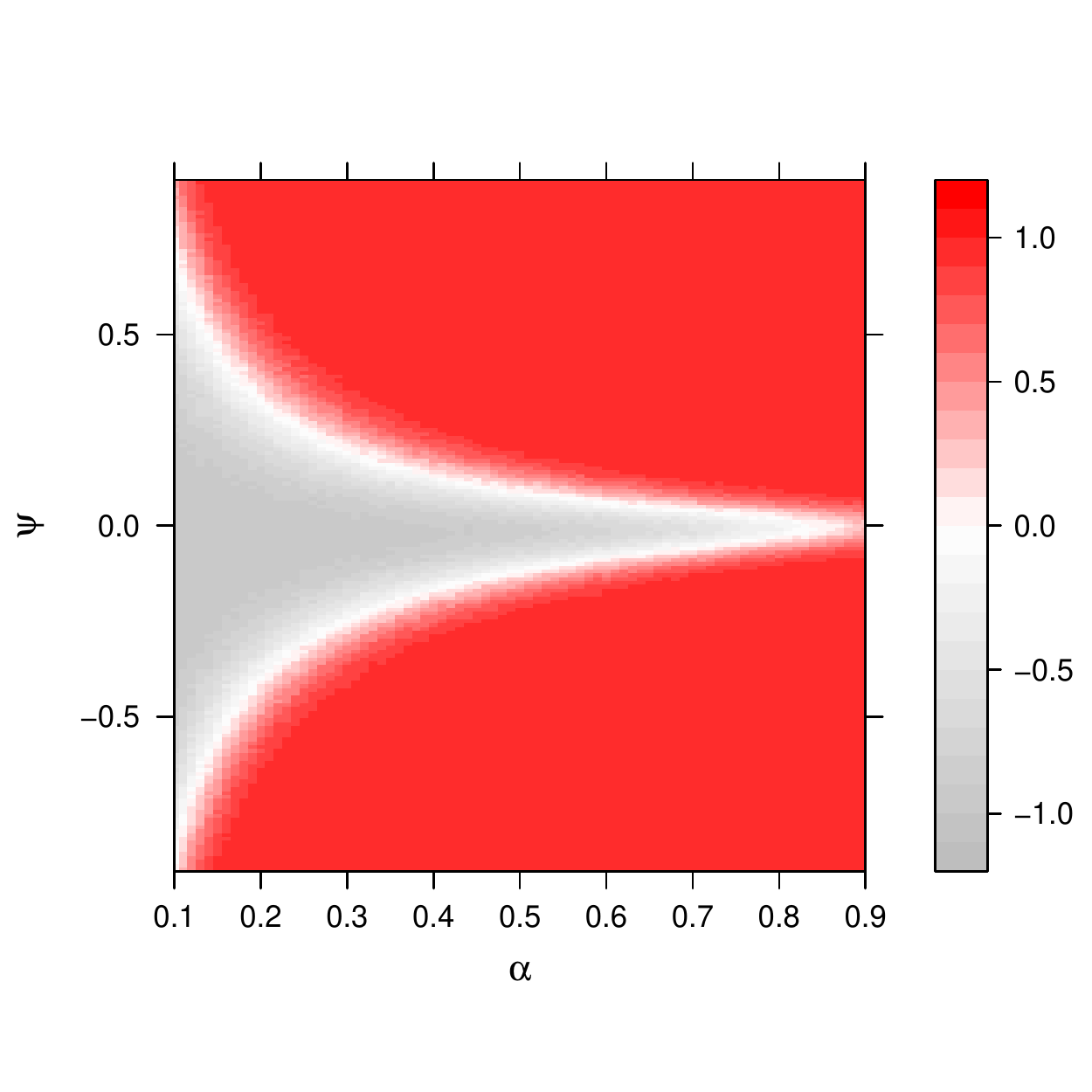}
  \vspace{-5mm}
  \caption{\footnotesize{Informative Sign - Weak Covariates}}
  \label{fig:sub4}
\end{subfigure}

\begin{subfigure}{.45\textwidth}
  \centering
  \includegraphics[width=.9\linewidth]{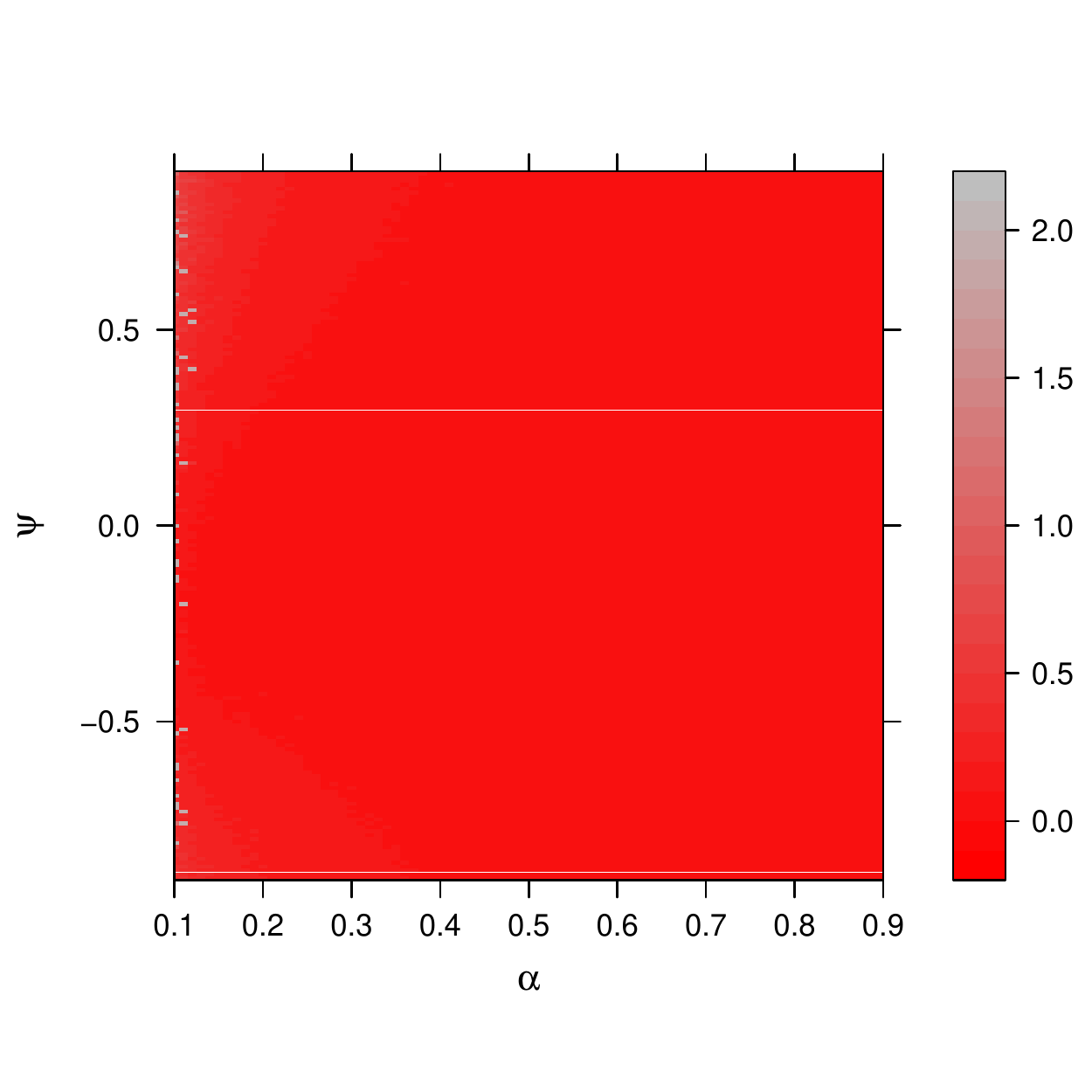}
  \vspace{-5mm}
  \caption{\footnotesize{Length of Bounds - Strong Covariates}}
  \label{fig:sub3}
\end{subfigure}
\begin{subfigure}{.45\textwidth}
  \centering
  \includegraphics[width=.9\linewidth]{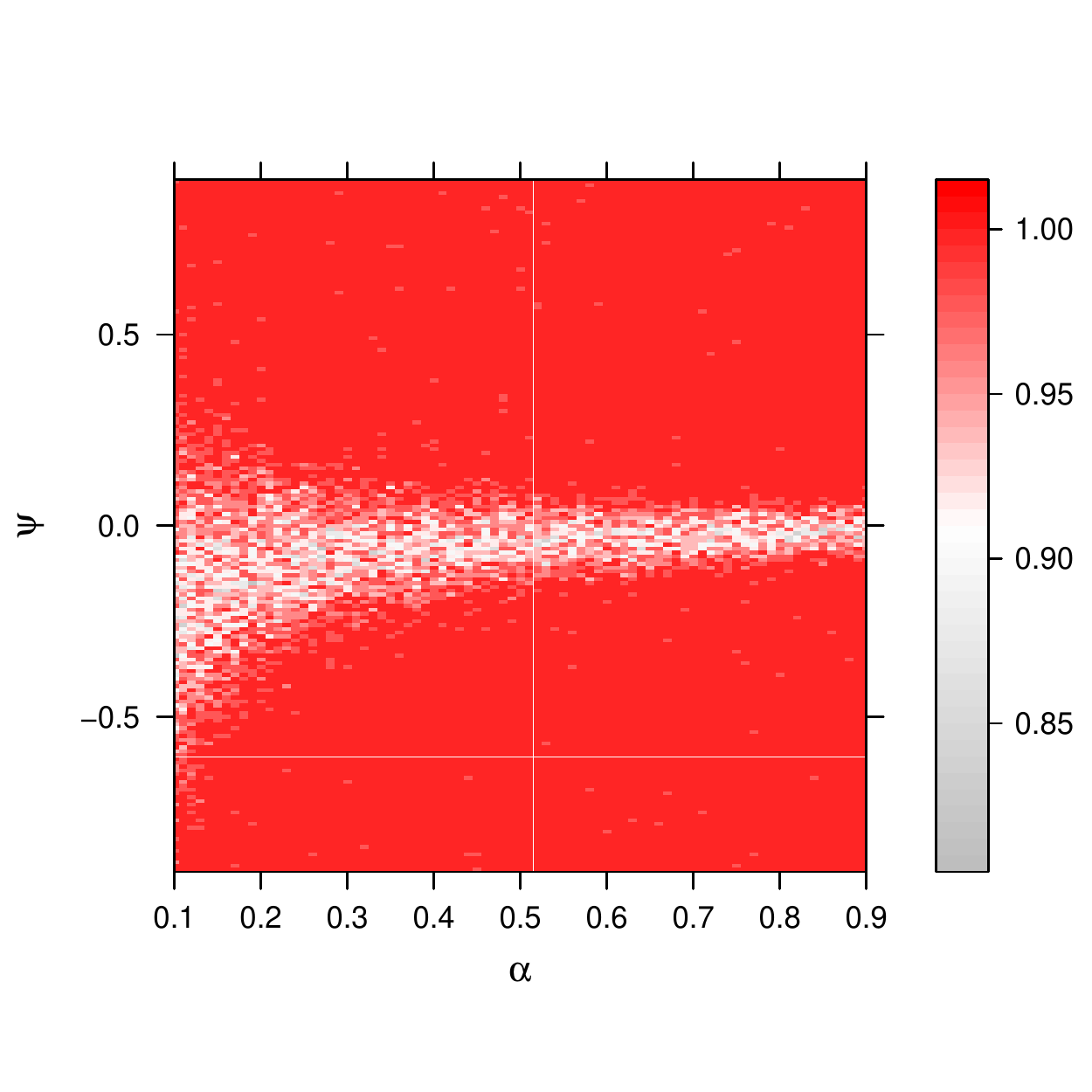}
  \vspace{-5mm}
  \caption{\footnotesize{Informative Sign - Strong Covariates}}
  \label{fig:sub4}
\end{subfigure}

\caption{\footnotesize{Estimated bounds using simulated data from a known causal model, across 1000 simulation for each $\alpha$ and $\psi$ value. Grey areas indicate that bounds are uninformative. ``Weak covariates'' refers to the case where covariates weakly predict $S = 1$. ``Strong covariates'' refers to the case where covariates strongly predict $S = 1$.}}
\label{fig:sim.plots}
\end{figure}

\noindent the average length of the bounds decreases to 0.58, and with the strongly predictive covariates, the average length of the bounds was 0.06. The simulations clearly indicate that the bounds can be highly informative, but it is essential that either values of $\alpha$ are fairly high or that the covariates be at least weakly predictive of $S$. In the ICU data, we found $\alpha$ to be much lower, and the covariates were not strongly associated with ICU admission. In these situations, use of the sensitivity analysis will likely be more informative.

\section{Discussion}
\label{sec:dis}

One critical clinical question of practical importance given limited health care resources is whether delayed admission to the ICU is beneficial. Answers to this question are elusive since the population referred to critical care tends to be sicker than the population that remains on a general hospital ward, and wait times for the ICU are censored for any patients not admitted to the ICU. While instruments can allow for causal inference in the presence of unobserved confounders, we proved that IV methods do not identify the usual causal estimand of interest in our application. We demonstrated that while  point-identification is not possible without further assumptions, partial identification is possible.  These bounds can incorporate covariate information in a flexible and data-adaptive way.  Moreover, our use of influence function-based methods allows for efficient estimation and inference, even when modeling covariate relationships very flexibly using machine learning tools. Despite their usefulness, influence-function-based approaches seem not to have been used previously in partial identification problems, perhaps due to the relative complexity and nonsmoothness issues discussed in Section 3.3. Simulations show that our bounds can be highly informative, but the instrument must encourage a relatively large portion of the population to be survivor-compliers, or else covariates must strongly predict selection. That is,  to produce bounds that are highly informative about the true causal effect, covariate selection is critical unless the effect of the instrument is very strong. Our results serve as an important guide for data collection, by indicating to investigators the importance of measuring covariates that predict selection. For example, one could utilize pilot study data to identify such covariates. Nonetheless our data-driven sensitivity analysis approach can provide a useful picture even when bounds are wide, as observed in our ICU application.

As is often the case with studies based on partial identification, we cannot offer the certainty of a point estimate and confidence interval. However, we think our proposed methods offer important evidence about the clinical question of ICU waiting times. While the bounds on the causal effect of waiting times were uninformative about the true causal effect, use of a sensitivity analysis did provide evidence that waiting less than four hours to receive critical care reduces mortality.  While we cannot rule out causal effects of the opposite sign, we argue that, taken with other analyses, there is some evidence that reducing waiting times for the ICU is a clinical goal worth pursuing. Although for some causal questions the best available evidence may not take the form of a single point estimate, we have showed that advances in nonparametric efficiency theory and data-adaptive regression methods can still play a crucial role. We look forward to more future developments in partial identification problems with interesting study designs and complex confounding.

\clearpage
\stepcounter{section}
\printbibliography[title={\thesection \ \ \ References}]

\pagebreak
\setcounter{page}{1}

\section{Appendix}

\setlength{\parindent}{0cm}

\subsection{Proof of Proposition 1}

To prove non-identifiability of $\psi$, we only need to find an observed data distribution that can yield different values of the parameter $\psi$. We give a simple example here, but non-identifiability also follows from the fact that the bounds in Theorem 1 are sharp. Suppose there are no covariates ($\bX=\emptyset$) and there are only two principal strata, indexed by variable $C \sim \text{Bernoulli}(0.5)$ with
\begin{align*}
C=1 &\iff (S^{(z=0)}=S^{(z=1)}=1, A^{(z=1)}>A^{(z=0)}), \\
C=0 &\iff (S^{(z=0)}<S^{(z=1)}, A^{(z=1)}=1) .
\end{align*}
This means there are only treatment-compliers and selection-compliers who take treatment. By unconfoundedness, $ \Pb(Z=z,S=s,A=a) = \Pb(S^{(z)}=s,A^{(z)}=a) = \one(z=a=1)/2$. Also suppose $Y$ is binary and $Y^{(z=0)}=0$ with probability one. Then
$\Pb(Y=1 \mid A,S, Z)=0$ except when $A=S=Z=1$, in which case
$$ \Pb(Y=1 \mid A=S=Z=1)= \sum_{c=0}^1 \Pb(Y^{(z=1)}=1 \mid C=c) \Pb(C=c) = 0.5 (\psi + \xi) $$
where $\xi=\Pb(Y^{(z=1)}=1 \mid C=0)$. Thus any choices of $(\psi,\xi)$ with the same sum $(\psi+\xi)$ would yield the same observed data distribution; for example, both $(\psi,\xi)=(1,0)$ and $(\psi,\xi)=(0,1)$ give $\Pb(Y=1 \mid A=S=Z=1)=0.5$. 

\subsection{Proof of Proposition 2}

We have
\begin{align*}
\beta &= \E(Y^{(z=1)} - Y^{(z=0)} \mid S^{(z=1)}=S^{(z=0)}=1) \\   
&= \E(Y^{(z=1,A^{(z=1)})} - Y^{(z=0,A^{(z=0)})} \mid S^{(z=1)}=S^{(z=0)}=1) \\
&= \E(Y^{(A^{(z=1)})} - Y^{(A^{(z=0)})} \mid S^{(z=1)}=S^{(z=0)}=1) \\
&= \E\{(Y^{(a=1)} - Y^{(a=0)}) \one(A^{(z=1)} > A^{(z=0)}) \mid S^{(z=1)}=S^{(z=0)}=1\} \\
&= \E(Y^{(a=1)} - Y^{(a=0)} \mid S^{(z=1)}=S^{(z=0)}=1, A^{(z=1)} > A^{(z=0)}) \\
& \hspace{.4in} \times \Pb(A^{(z=1)} > A^{(z=0)} \mid S^{(z=1)}=S^{(z=0)}=1) \\
&= \psi \alpha / \Pb(S^{(z=1)}=S^{(z=0)}=1)
\end{align*}
where the first equality follows by definition, the second by the fact that $Y^{(z)}=Y^{(z,A^{(z)})}$ from Assumption 1 (consistency), the second by Assumption 5 (exclusion), the third by Assumption 6 (monotonicity), the fourth by iterated expectation and Assumption 3 (instrumentation), and the last by definition. Rearranging gives the desired result.

\subsection{Proof of Theorem 1}

To ease notation, in this section all potential outcomes are with respect to interventions on the instrument $Z$ (not treatment $A$), so that we can write $Y^1 = Y^{(z=1)}$, for example. 

\subsubsection{Bounds on $\alpha$}

For the upper bound on $\alpha$, clearly we have
$$ \Pb(A^1 > A^0 \mid \bX) \leq \Pb(A^1 > A^0 \mid \bX) + \Pb(A^1 > S^0 \mid \bX) . $$
Now note that
\begin{align*} \{A^1 > A^0\} \cup \{A^1 > S^0\} &= [ \{A^1 > A^0\} \cup \{A^1 > S^0\} \cup \{A^1=A^0=1\}] \cap \{A^1=A^0=1\} \\
&= \{A^1=1\} \cap \{A^1=A^0=1\} = \{A^1=1\} \cap \{A^0=1\}
\end{align*}
where the first equality follows from simple logic, the second by definition of $\{A^1=1\}$, and the third by monotonicity. Therefore
\begin{align*} \Pb(A^1 > A^0 \mid \bX) &\leq \Pb(A^1 = 1 \mid \bX) - \Pb(A^0 = 1 \mid \bX) \\
&= \Pb(A = 1 \mid \bX,Z=1) - \Pb(A = 1 \mid \bX,Z=0)
\end{align*}
where the equality follows by consistency, positivity, and unconfoundedness. Hence
\begin{align*}
\alpha &= \E\{ \Pb(A^1 > A^0 \mid \bX)\} \\
&\leq \E\{ \Pb(A = 1 \mid \bX,Z=1) - \Pb(A = 1 \mid \bX,Z=0) \} = \alpha_u .
\end{align*}

For the lower bound on $\alpha$, we similarly have
$$ \Pb(A^1 > A^0 \mid \bX) \geq \{ \Pb(A^1 > A^0 \mid \bX) - \Pb(S^0 =A^1=0 \mid \bX) \}_+ . $$
Now note that the right-hand side (before taking the positive part) is
\begin{align*} 
\Pb(A^1 > A^0 \mid \bX) &+ \Pb(A^0=A^1=0 \mid \bX) - \Pb(A^0=A^1=0 \mid \bX) - \Pb(S^0 =A^1=0 \mid \bX) \\
&= \Pb(A^0=0 \mid \bX) - \Pb(A^1 = 0 \mid \bX) \\
&= \Pb(A=0 \mid \bX, Z=0) - \Pb(A = 0 \mid \bX, Z=1)
\end{align*}
where the first equality follows by the facts that $\{A^1 > A^0\} \cup \{A^0=A^1=0\} = \{A^0=0\}$ since $A$ is binary and
$$ \{ A^0=A^1=0 \} \cup \{ S^0 =A^1=0 \} = \{A^0=0\} $$
by monotonicity, and the second follows by consistency, positivity, and unconfoundedness. Hence
\begin{align*}
\alpha &= \E\{ \Pb(A^1 > A^0 \mid \bX)\} \geq \E[\{ \Pb(A =0 \mid \bX,Z=0) - \Pb(A =0 \mid \bX,Z=1) \}_+] = \alpha_\ell .
\end{align*}

\subsubsection{Bounds on $\beta$}

Before deriving the bounds on $\beta$, we first give a useful lemma.

\begin{lemma}
\label{lem:mix_bds}
Suppose $H=pF + qG$, where $(H,F,G)$ are cumulative distribution functions for non-negative random variables $(X_H,X_F,X_G)$, and $p>0$ and $q>0$ are weights with $p+q=1$. Suppose the parent distribution $H$ and weights $p$ and $q$ are known, but the component distributions $F$ and $G$ are unknown.  Then sharp bounds on the mean under $F$ are given by
$$ \int \frac{(p-H) \vee 0}{p} = \int (1-F_\ell^* ) \leq \int X_F \ dF \leq \int (1-F_u^*) = \int \frac{(1-H) \wedge p}{p} $$ 
for the bounding distributions $F_\ell^* = \left( \frac{H}{p} \right) \wedge 1$ and $F_u^* = \left( \frac{H-q}{p} \right) \vee 0$.
\end{lemma}

\begin{proof}
Since the random variables associated with the distributions $(H,F,G)$ are non-negative, we can write expectations as integrated survival functions, as in $\int X_F \ dF = \int (1 - F)$. Therefore to show that the means under $F_\ell^*$ and $F_u^*$ are valid bounds, we must show that $\int F \leq \int F_\ell^*$ and $\int F_u^* \leq \int F$. \\

For $F_\ell^*$ note that
\begin{align*}
\int (F_\ell^* - F) &= \int \left( \frac{H}{p} \right) \wedge 1 - \int \left( \frac{ H - qG}{p} \right)  \\
&= \int_{H>p} (1-F) + \int_{H \leq p} \frac{qG}{p} \geq 0 
\end{align*}
where the last inequality follows since $1 \geq F=(H-qG)/p$ and $qG/p \geq 0$, because $(F,G,p,q)$ are all probabilities bounded between zero and one. Similarly, for $F_u^*$ we have
\begin{align*}
\int (F - F_u^*) &= \int \left( \frac{ H - qG}{p} \right) - \int \left( \frac{H-q}{p} \right) \vee 0 \\
&= \int_{H > q} \frac{q (1-G)}{p}  + \int_{H \leq q} F \geq 0
\end{align*}
where again the last inequality follows since $(F,G,p,q)$ are all in $[0,1]$. \\

To show sharpness, we must give component distributions $F$ and $G$ that attain the bounds and can be mixed using $(p,q)$ to form any known $H$, i.e., we can show that $pF_u^* + q G_\ell^*=H$ for $G_\ell^*=(H/q)\wedge 1$ (for $F_\ell^*$ we can simply reverse the role of $F$ and $G$). This follows since
\begin{align*}
pF_u^* + qG_\ell^* &= p \left\{ \left( \frac{H-q}{p} \right) \vee 0 \right\} + q \left\{ \left( \frac{H}{q} \right) \wedge 1 \right\} \\
&= \begin{cases} p \left( \frac{H-q}{p} \right) + q & \text{ if } H > q \\
0 +  q  \left( \frac{H}{q} \right) & \text{ if } H \leq q \end{cases} \\
&= H .
\end{align*}
\end{proof}

First note that 
\begin{align*}
\beta &= \E\{\E(Y^1 - Y^0 \mid \bX, S^0=S^1 = 1) \Pb(S^0=S^1=1 \mid \bX) \} / \Pb(S^0=S^1=1) \\
&=  \E\{\E(Y^1 - Y^0 \mid \bX, S^0 = 1) \Pb(S^0=1 \mid \bX) \} / \Pb(S^0=1) \\
&= \frac{\E[\{\E(Y^1 \mid \bX, S^0 = 1) - \E(Y \mid \bX, S=1, Z=0) \} \Pb(S=1 \mid \bX, Z=0) ] }{ \E\{ \Pb(S=1 \mid \bX, Z=0) \} } 
\end{align*} 
where the first equality follows by definition (and iterated expectation), the second by monotonicity, and the third by consistency, positivity, and unconfoundedness. \\

Now we use Lemma 1 to construct bounds for $\E(Y^1 \mid \bX, S^0 = 1)$. Let
\begin{align*}
H &= \Pb(Y^1 \leq y \mid \bX, S^1=1) \\
F &= \Pb(Y^1 \leq y \mid \bX, S^0=1) \\
p &= \Pb(S^0=1 \mid \bX, S^1=1) 
\end{align*}
so that $\E(Y^1 \mid \bX, S^0=1) = \int (1-F) =1-F(0)=\overline{F}(0)$ since $Y$ is binary. Also note that
\begin{align*}
H & =\Pb(Y \leq y \mid \bX, S=1, Z=1) \\
p & = \frac{\Pb(S^0=1 \mid \bX)}{\Pb(S^1=1 \mid \bX)} = \frac{\Pb(S=1 \mid \bX, Z=0)}{\Pb(S=1 \mid \bX, Z=1)} = \frac{\lambda_0(\bX)}{\lambda_1(\bX)}
\end{align*}
where the equality for $H$ and the second equality for $p$ use consistency, positivity, and unconfoundedness, and the first equality for $p$ uses monotonicity. \\

Therefore Lemma 1 gives
\begin{align*}
\E(Y^1 \mid \bX, S^0 = 1) &\geq \int \frac{(p-H) \vee 0}{p} = \int \frac{(\overline{H} - q) \vee 0}{p} \\
&= \{ \mu_1(\bX) - \lambda_1(\bX) + \lambda_0(\bX) \}_+ / \lambda_0(\bX)
\end{align*}
where the second equality comes from rearranging and using the fact that $Y$ is binary so that $\int \overline{H} = \overline{H}(0) = \E(Y \mid \bX, S=1, Z=1)$. Similarly we have
\begin{align*}
\E(Y^1 \mid \bX, S^0 = 1) &\leq \int \frac{(1-H) \wedge p}{p} \\
&= \{ \mu_1(\bX) \wedge \lambda_0(\bX) \} / \lambda_0(\bX) .
\end{align*}

Therefore, using the definitions from the main text and plugging in the results from Lemma 1 gives
\begin{align*}
\beta &= \E\{ \E(Y^1 \mid \bX, S^0 = 1) \lambda_0(\bX) - \mu_0(\bX) \} / \E\{ \lambda_0(\bX) \} \\
& \geq \E\{ \{ \mu_1(\bX) - \lambda_1(\bX) + \lambda_0(\bX) \}_+ - \mu_0(\bX) \} / \E\{ \lambda_0(\bX) \} = \beta_\ell
\end{align*}
and similarly
\begin{align*}
\beta &\leq \E\{ \{ \mu_1(\bX) \wedge \lambda_0(\bX) \} - \mu_0(\bX) \} / \E\{ \lambda_0(\bX) \} = \beta_u
\end{align*}

\subsection{Proof of Theorem 2}

Theorem 2 follows from the bounds on $\alpha$ and $\beta$ given in Theorem 1, along with the expression for $\psi$ given in Proposition 2. In particular, note that the bounds on $\beta$ imply
$$ \beta_\ell\E\{ \lambda_0(\bX)\}  \leq \beta \E\{ \lambda_0(\bX)\} \leq \beta_u \E\{ \lambda_0(\bX)\} $$
where we used the fact that 
$$ \Pb(S^0=S^1=1 ) = \Pb(S^0=1) = \E\{ \Pb(S=1 \mid \bX, Z=0)\} > 0 $$
where the first equality comes from monotonicity, and the second from consistency, positivity, and unconfoundedness. \\

Now we consider three cases depending on whether the above bounds on the numerator are positive or negative. All three results follow from the fact that $c>0$ and $0 < \alpha_\ell \leq \alpha \leq \alpha_u \leq 1$, we have $c/\alpha_u \leq c/\alpha_\ell$ and $-c/\alpha_u \geq -c/\alpha_\ell$. If both numerator bounds are non-negative so that $0 \leq \beta_\ell$ then
$$ \frac{ \beta_\ell\E\{ \lambda_0(\bX)\} }{\alpha_u} \leq \frac{ \beta_\ell\E\{ \lambda_0(\bX)\} }{\alpha} \leq \psi \leq \frac{\beta_u\E\{ \lambda_0(\bX)\} }{\alpha} \leq \frac{\beta_u\E\{ \lambda_0(\bX)\} }{\alpha_\ell} . $$

Similarly if both numerator bounds are zero or negative so that $\beta_u \leq 0$ then
$$ \frac{ \beta_\ell\E\{ \lambda_0(\bX)\} }{\alpha_\ell} \leq \frac{ \beta_\ell\E\{ \lambda_0(\bX)\} }{\alpha} \leq \psi \leq \frac{\beta_u\E\{ \lambda_0(\bX)\} }{\alpha} \leq \frac{\beta_u\E\{ \lambda_0(\bX)\} }{\alpha_u} . $$

Finally if the lower numerator bound is non-positive and the upper is non-negative so that $\beta_\ell  \leq 0 \leq \beta_u $ then it follows that
$$ \frac{ \beta_\ell\E\{ \lambda_0(\bX)\} }{\alpha_\ell} \leq \frac{ \beta_\ell\E\{ \lambda_0(\bX)\} }{\alpha} \leq \psi \leq \frac{\beta_u\E\{ \lambda_0(\bX)\} }{\alpha} \leq \frac{\beta_u\E\{ \lambda_0(\bX)\} }{\alpha_\ell} . $$

This yields the desired result.

\subsection{Proof of identification when $\bX$ predicts $S$}

Here we consider the case where the covariates $\bX$ can perfectly predict selection $S$, in the sense that there exists some mapping $g: \mathcal{X} \mapsto \{0,1\}$ such that $S=g(\bX)$ wp1. We will show that in this case the bounds collapse to a single point, and the SCATE $\psi$ is identified. \\

First consider the bounds on $\beta$. Note that if $S=g(\bX)$ then
$$ \lambda_Z(\bX) = \E(S \mid \bX, Z) = \E\{g(\bX) \mid \bX, Z\} =g(\bX) = S . $$
Therefore $\lambda_0(\bX)-\lambda_1(\bX)=S-S=0$ and the lower bound on $\beta$ is given by
$$ \beta_\ell = \frac{ \E[\{ \mu_1(\bX) + \lambda_0(\bX) - \lambda_1(\bX) \}_+ - \mu_0(\bX)] }{ \E\{ \lambda_0(\bX) \}} =  \frac{ \E\{ \mu_1(\bX) - \mu_0(\bX) \} }{ \E\{ \lambda_0(\bX) \}} . $$
Similarly 
$$ \mu_1(\bX) \wedge \lambda_0(\bX) = \mu_1(\bX) \wedge S = \mu_1(\bX) $$
because 
$$ \mu_Z(\bX) = \E(SY \mid \bX, Z) = \E\{ g(\bX) Y \mid \bX, Z \} = g(\bX) \E(Y \mid \bX, Z) = S\E(Y \mid \bX, Z) \leq S  $$
since $Y \in [0,1]$. Therefore the upper bound on $\beta$ is given by
$$ \beta_u = \frac{ \E[ \{  \mu_1(\bX) \wedge \lambda_0(\bX) \} - \mu_0(\bX) ]}{\E\{ \lambda_0(\bX)\}} = \frac{ \E\{ \mu_1(\bX) - \mu_0(\bX) \} }{ \E\{ \lambda_0(\bX) \} } $$
and thus matches the lower bound, so that $\beta$ is point-identified. \\

Now consider the bounds on $\alpha$. Recall $\theta_z(a \mid \bX) = \Pb(A^z=a \mid \bX)$ based on the identifying assumptions. If we also have $S=g(\bX)$ then 
\begin{align*}
\theta_z(a \mid \bX) &= \Pb\{A=a \mid \bX, g(\bX), Z=z\} = \Pb(A=a \mid \bX, S,Z=z) \\
&= \Pb(A^z=a \mid \bX, S^z)
\end{align*}
so that $S=g(\bX)$ implies $A^z \ind S^z \mid \bX$. Hence
$$ \theta_z(a \mid \bX) = \Pb(A^z=a \mid \bX, S^z=1) . $$
But $A^z \in \{0,1\}$ when $S^z=1$, so that 
$$ \theta_1(1 \mid \bX) - \theta_0(1 \mid \bX) = \{1 - \theta_1(0 \mid \bX) \} - \{ 1 - \theta_0(0 \mid \bX) \} $$
and therefore the lower and upper bounds $\alpha_\ell$ and $\alpha_u$ are equal, so that $\alpha$ is point-identified.

\subsection{Proof of Theorem 3}

\subsubsection{Efficient influence function for $\psi^*(\bdelta)$}

First we prove two lemmas describing efficient influence functions for bounds on $(\alpha,\beta)$.

\begin{lemma}
Under Condition 1 the efficient influence functions for the bounds $(\alpha_\ell,\alpha_u)$ on the survivor-complier proportion are given by $\varphi^{(\alpha)}_{\ell/u} - \alpha_{\ell/u}$ for
\begin{equation*}
\begin{gathered}
\varphi^{(\alpha)}_\ell = \one\Big\{ \theta_0(0 \mid \bX)> \theta_1(0 \mid \bX) \Big\} \Big[ \phi_1\{\one(A \neq 0)\} - \phi_0\{\one(A \neq 0)\} \Big]  \\
\varphi^{(\alpha)}_u = \phi_1\{\one(A =1)\} - \phi_0\{\one(A = 1)\}.
\end{gathered}
\end{equation*}
\end{lemma}

\begin{proof}
The fact that $\varphi_u^{(\alpha)} - \alpha_u$ is the efficient influence function for $\alpha_u$ is well-known, since $\alpha_u$ is mathematically equivalent to the $\bX$-adjusted average treatment effect of $Z$ on $\one(A=1)$. This result has been previously discussed by \textcite{robins1994estimation, hahn1998role, scharfstein1999adjusting} among others. Thus we need only consider the parameter $\alpha_\ell$. \\

The lower bound $\alpha_\ell$ is more delicate, in particular since it is a non-smooth function of the nuisance functions $\theta_z(0 \mid \bX)$. Nonetheless we can adapt a result from the optimal treatment regime literature to give conditions under which it is pathwise differentiable (i.e., has an influence function) and regularly estimable at $\sqrt{n}$ rates. \\

Letting $\gamma_1(\bX) = \theta_1(0 \mid \bX) - \theta_0(0 \mid \bX)$, note that 
$$ \alpha_\ell = \E\{ \gamma_1(\bX)_+ \} = \E[\gamma_1(\bX) \one\{ \gamma_1(\bX) > 0 \} ] . $$
Now we can use the same logic as in Theorem 2 and Lemma 2 of \textcite{van2014targeted} to show that the influence function for $\alpha_\ell$ under Condition 1 is the same as that treating the indicator $\one\{\gamma_1(\bX) > 0\}$ in $\alpha_\ell$ as known. Hence under Condition 1 the non-smoothness of $\alpha_\ell$ is inconsequential. \\

Specifically, letting $\{ \Pb_\epsilon : \epsilon \in \R\}$ denote a smooth parametric submodel passing through $\Pb$ at $\epsilon=0$ (e.g., $d\Pb_\epsilon=(1 + \epsilon h) d\Pb$ for bounded mean-zero $h=h(\bO)$), we have (suppressing the subscript on $\gamma_1$ for simplicity)
\begin{align*}
\alpha_\ell(\Pb_\epsilon) - \alpha_\ell(\Pb) &= \int \gamma_\epsilon \one(\gamma_\epsilon > 0) \ d\Pb_\epsilon - \int \gamma \one(\gamma > 0) \ d\Pb \\
&= \int \one(\gamma > 0) (\gamma_\epsilon \ d\Pb_\epsilon - \gamma \ d\Pb) + \int \gamma_\epsilon \Big\{ \one(\gamma_\epsilon > 0) - \one(\gamma > 0) \Big\} \ d\Pb_\epsilon .
\end{align*}
The first term, after dividing by $\epsilon$ and letting $\epsilon \rightarrow 0$, is the pathwise derivative for the case where the indicator $\one(\gamma_1 > 0)$ is known. This corresponds to the influence function $\varphi_\ell^{(\alpha)}$ given in the statement of Lemma 2, which follows by standard chain rule arguments as in \textcite{hahn1998role} and elsewhere. Now, again following the same logic as in Lemma 2 of \textcite{van2014targeted}, we will show that the second term is $o(|\epsilon|)$ under Condition 1, and so does not contribute to the influence function. \\

In absolute value, the second term is bounded above by
\begin{align*}
\int  | \gamma_\epsilon | \Big|  \one &(\gamma_\epsilon > 0) - \one(\gamma > 0) \Big| \ d\Pb_\epsilon \leq \int | \gamma_\epsilon | \one(|\gamma| < | \gamma_\epsilon - \gamma | )  \ d\Pb_\epsilon \\
&\leq \int (| \gamma| + C |\epsilon| ) \one(|\gamma| < C|\epsilon| )  \ d\Pb_\epsilon \lesssim |\epsilon| (1 + |\epsilon|) \int \one(|\gamma| < C|\epsilon| )  \ d\Pb \\
&= |\epsilon| (1 + |\epsilon|) \Pb(0<|\gamma| < C|\epsilon| ) = o(|\epsilon|) 
\end{align*}
where the second bound follows since $|\gamma| = | \gamma_\epsilon - \gamma| - |\gamma_\epsilon| $ whenever $\gamma_\epsilon$ and $\gamma$ have different signs, the third and the fourth by the submodel construction and boundedness of $\gamma$, and the fifth by Condition 1. Hence this term does not contribute to the pathwise derivative, and the influence function for $\alpha_\ell$ is the same as if the indicator $\one(\gamma_1>0)$ was known.
\end{proof}

\begin{lemma}
Under Condition 2 the efficient influence functions for the bounds $(\beta_\ell,\beta_u)$ on the survivor intention-to-treat effect are given by $\E\{\phi_0(S)\}^{-1} \{ \varphi^{(\beta)}_{\ell/u}  - \beta_{\ell/u} \phi_0(S) \}$ for
\begin{equation*}
\begin{gathered}
\varphi^{(\beta)}_\ell = \one\left\{ \frac{\mu_1(\bX)}{\lambda_1(\bX)-\lambda_0(\bX)} > 1 \right\} \Big[ \phi_1\{S(Y-1)\} - \phi_0(S) \Big] - \phi_0(SY)   \\
\varphi^{(\beta)}_u = \one\left\{  \frac{\mu_1(\bX)}{\lambda_0(\bX)} > 1 \right\} \Big\{ \phi_0(S) - \phi_1(SY) \Big\} + \phi_1(SY) - \phi_0(SY) .
\end{gathered}
\end{equation*}
\end{lemma}

\begin{proof}
This proof is similar to that of Lemma 2. The efficient influence function for the denominator $\E\{\lambda_0(\bX)\}$ of the bounds on $\beta$ is straightforward, as this parameter is mathematically equivalent to the marginal mean of an outcome missing at random \autocite{robins1994estimation, hahn1998role, scharfstein1999adjusting}. Specifically the influence function is given by $\phi_0(S)-\E\{\lambda_0(\bX)\}$. The same goes for the subtracted numerator term $\E\{\mu_0(\bX)\}$, which similarly has influence function $\phi_0(SY)-\E\{\mu_0(\bX)\}$.  \\

Now consider the non-smooth terms in the numerators. Letting $\gamma_2(\bX) = \mu_1(\bX) + \lambda_0(\bX) - \lambda_1(\bX)$, we have that the non-smooth term in the numerator of $\beta_\ell$ is
$$ \E\{ \gamma_2(\bX)_+\} = \E[\gamma_2(\bX) \one\{\gamma_2(\bX)>0\}]$$
and so can be analyzed exactly as in Lemma 2, except replacing Condition 1 with 
$$ \Pb\{ \mu_1(\bX) = \lambda_1(\bX)-\lambda_0(\bX)\}=0 $$
as in Condition 2, to ensure that $\gamma_2=\mu_1+\lambda_0-\lambda_1$ does not have a point mass at zero. Therefore the influence function for this parameter is the same as that treating the indicator $\one(\mu_1 > \lambda_1-\lambda_0)$ as known, which is given in the statement of Lemma 3. \\

Similarly, now letting $\gamma_3(\bX) = \mu_1(\bX)-\lambda_0(\bX)$, the non-smooth term in the numerator of $\beta_u$ is given by
\begin{align*} 
\E\{ \mu_1(\bX) \wedge \lambda_0(\bX) \} &= \E[ \mu_1(\bX) \one\{\lambda_0(\bX) \geq \mu_1(\bX)\} + \lambda_0(\bX) \one\{\mu_1(\bX) > \lambda_0(\bX) \}] \\
&= \E[ \mu_1(\bX) + \{\lambda_0(\bX)-\mu_1(\bX)\} \one\{\mu_1(\bX) > \lambda_0(\bX) \}] \\
&= \E[ \mu_1(\bX) - \gamma_3(\bX) \one\{ \gamma_3(\bX) > 0 \} ] .
\end{align*}
Again the first term above can be analyzed with standard techniques, as it is mathematically equivalent to the marginal mean of an outcome missing at random. The second term is exactly the same as in the previous two examples in Lemma 2 and the first part of this Lemma 3. The second part of Condition 2, that
$$ \Pb\{ \mu_1(\bX) = \lambda_0(\bX)\}=0 , $$
again ensures that $\gamma_3=\mu_1-\lambda_0$ does not have a point mass at zero. Therefore the influence function is the same as that treating the indicator $\one(\gamma_3>0)=\one(\mu_1>\lambda_0)$ as known. \\

Now that we have the influence functions for each the components making up the lower and upper bounds of $\beta$, the final result of Lemma 3 follows after combining these influence functions using the chain rule, and rearranging.
\end{proof}

That the efficient influence function of $\psi^*(\bdelta)$ is given as in Theorem 3, i.e., 
$$ \varphi(\boldsymbol\eta) = \{\varphi^{(\beta)}(\delta_2) - \psi^*(\bdelta) \varphi^{(\alpha)}(\delta_1) \} / {\E\{\varphi^{(\alpha)}(\delta_1)\} } $$
with $\varphi^{(\alpha)}(\delta_1) = \delta_1 \varphi_u^{(\alpha)} + (1-\delta_1) \varphi_\ell^{(\alpha)}$ and $\varphi^{(\beta)}(\delta_2) = \delta_2 \varphi_u^{(\beta)} + (1-\delta_2) \varphi_\ell^{(\beta)}$ 
now follows directly from Lemmas 2--3 together with the chain rule. 

\subsubsection{Asymptotic results for $\widehat\psi^*(\bdelta)$}

To ease notation, in this subsection we drop the dependence of all quantities on $\bdelta$, and write 
$$ \hat\psi = \widehat\psi^*(\bdelta) \ , \ \varphi_\alpha = \varphi^{(\alpha)}(\delta_1) \ , \ \varphi_\beta = \varphi^{(\beta)}(\delta_2) . $$
By definition, the estimator $\hat\psi$ solves the efficient influence function estimating equation, i.e., 
$\hat\psi = { \Pn\{ \varphi_\beta(\hat{\boldsymbol\eta}) \} }/ { \Pn\{ \varphi_\alpha(\hat{\boldsymbol\eta}) \} }$. 
Therefore we have
\begin{align*}
\hat\psi - \psi &=  \frac{ \Pn\{ \varphi_\beta(\hat{\boldsymbol\eta}) \} }{ \Pn\{ \varphi_\alpha(\hat{\boldsymbol\eta}) \} } - \frac{\Pb \{ \varphi_\beta(\boldsymbol\eta)\} }{ \Pb\{ \varphi_\alpha(\boldsymbol\eta)\} }  = \frac{ \Pb\{ \varphi_\alpha(\boldsymbol\eta)\} \Pn\{ \varphi_\beta(\hat{\boldsymbol\eta}) \}   - \Pb \{ \varphi_\beta(\boldsymbol\eta)\} \Pn\{ \varphi_\alpha(\hat{\boldsymbol\eta}) \}  }{ \Pn\{ \varphi_\alpha(\hat{\boldsymbol\eta}) \} \Pb\{ \varphi_\alpha(\boldsymbol\eta)\} } \\
&= \frac{ \Pb\{ \varphi_\alpha(\boldsymbol\eta)\} [ \Pn\{ \varphi_\beta(\hat{\boldsymbol\eta}) \} - \Pb\{ \varphi_\beta({\boldsymbol\eta}) \} ]  - \Pb \{ \varphi_\beta(\boldsymbol\eta)\} [ \Pn\{ \varphi_\alpha(\hat{\boldsymbol\eta}) \} - \Pb\{ \varphi_\alpha({\boldsymbol\eta}) \} ]  }{ \Pn\{ \varphi_\alpha(\hat{\boldsymbol\eta}) \} \Pb\{ \varphi_\alpha(\boldsymbol\eta)\} } \\
&= \Pn\{ \varphi_\alpha(\hat{\boldsymbol\eta}) \}^{-1} \left( \Big[ \Pn\{ \varphi_\beta(\hat{\boldsymbol\eta}) \} - \Pb\{ \varphi_\beta({\boldsymbol\eta}) \} \Big] - \psi \Big[ \Pn\{ \varphi_\alpha(\hat{\boldsymbol\eta}) \} - \Pb\{ \varphi_\alpha({\boldsymbol\eta}) \} \Big] \right) 
\end{align*}

Now we will analyze the two terms in parentheses. For the $\beta$ part we have 
\begin{align*}
\Pn\{ \varphi_\beta(\hat{\boldsymbol\eta}) \} - \Pb\{ \varphi_\beta({\boldsymbol\eta}) \} &= (\Pn-\Pb)\{ \varphi_\beta(\hat{\boldsymbol\eta}) - \varphi_\beta({\boldsymbol\eta}) \} + (\Pn-\Pb)  \varphi_\beta({\boldsymbol\eta}) - \Pb\{ \varphi_\beta(\hat{\boldsymbol\eta}) -  \varphi_\beta({\boldsymbol\eta}) \} 
\end{align*}
The first term is a centered empirical process and is $o_\Pb(1/\sqrt{n})$ by, for example, Lemma 19.24 in \textcite{van2000asymptotic} since $\varphi_\beta$ lies in a Donsker class and $\|\varphi_\beta(\hat{\boldsymbol\eta}) - \varphi_\beta({\boldsymbol\eta}) \|^2=o_\Pb(1)$ by the assumptions of Theorem 3. With sample splitting this term will be $o_\Pb(1/\sqrt{n})$ under only the consistency assumption, without requiring any Donsker conditions. The second term is asymptotically normal after scaling by $\sqrt{n}$, by the central limit theorem. \\

The third term captures the effect of nuisance estimation, and for it we have
\begin{align*}
\Pb\{ \varphi_\beta(\hat{\boldsymbol\eta}) -  \varphi_\beta({\boldsymbol\eta}) \} &= \delta_2 \Pb\{ \varphi_{\beta,u}(\hat{\boldsymbol\eta}) -  \varphi_{\beta,u}({\boldsymbol\eta}) \} + (1-\delta_2) \Pb\{ \varphi_{\beta,\ell}(\hat{\boldsymbol\eta}) -  \varphi_{\beta,\ell}({\boldsymbol\eta}) \} .
\end{align*}
First consider the first term on the right hand side, referring to the definition of $\varphi_\beta$ in Lemma 3. Repeated iterated expectation shows that
$$ \Pb\{ \hat\phi_z(SY) - \phi_z(SY) \} \lesssim \| \hat\pi_1 - \pi_1 \| \| \hat\mu_z - \mu_z \| . $$
Let $\gamma_3=\mu_1-\lambda_0$ as before and also let $\phi = \phi_0(S)-\phi_1(SY)$. Then the same arguments, combined with the result from Theorem 3 of \textcite{van2014targeted}, show that
\begin{align*} 
\Pb\{\one(\hat\gamma_3>0) \hat\phi &- \one(\gamma_3>0) \phi \} = \Pb\{\one(\hat\gamma_3>0) (\hat\phi-\phi) + \{\one(\hat\gamma_3>0) - \one(\gamma_3>0)\} \phi \} \\
&\lesssim \| \hat\pi_1 - \pi_1 \| \Big( \| \hat\mu_1 - \mu_1 \| +  \|\hat\lambda_0 - \lambda_0 \| \Big) + \| \hat\gamma_3 - \gamma_3 \| \sqrt{ \Pb(\gamma_3 < | \hat\gamma_3 - \gamma_3 | ) } .
\end{align*}
Therefore, combining the above results, we have that $\Pb\{ \varphi_{\beta,u}(\hat{\boldsymbol\eta}) -  \varphi_{\beta,u}({\boldsymbol\eta}) \}$ is bounded above (up to constants) by
\begin{align*}
\| \hat\pi_1 - \pi_1 \| \Big( \max_z \| \hat\mu_z - \mu_z \| +  \|\hat\lambda_0 - \lambda_0 \| \Big) + \| \hat\gamma_3 - \gamma_3 \| \sqrt{ \Pb(\gamma_3 < | \hat\gamma_3 - \gamma_3 | ) }
\end{align*}
and this is $o_\Pb(1/\sqrt{n})$ by Assumptions 3--4 of Theorem 3. \\

The same logic shows that, for $\gamma_2 = \mu_1 - (\lambda_1-\lambda_0)$, 
\begin{align*}
\Pb\{ \varphi_{\beta,\ell}(\hat{\boldsymbol\eta}) -  \varphi_{\beta,\ell}({\boldsymbol\eta}) \} \lesssim \| \hat\pi_1 - \pi_1 \| &\Big( \max_z  \|\hat\lambda_z - \lambda_z \| +  \max_z  \| \hat\mu_z - \mu_z \| \Big) \\
& + \| \hat\gamma_2 - \gamma_2 \| \sqrt{ \Pb(\gamma_2 < | \hat\gamma_2 - \gamma_2 | ) },
\end{align*}
and this is also $o_\Pb(1/\sqrt{n})$ by Assumptions 3--4 of Theorem 3. \\

Therefore $\Pb\{ \varphi_\beta(\hat{\boldsymbol\eta}) -  \varphi_\beta({\boldsymbol\eta}) \} = o_\Pb(1/\sqrt{n})$, which implies
$$ \Pn\{ \varphi_\beta(\hat{\boldsymbol\eta}) \} - \Pb\{ \varphi_\beta({\boldsymbol\eta}) \} = (\Pn-\Pb)  \varphi_\beta({\boldsymbol\eta}) + o_\Pb(1/\sqrt{n}) . $$ 

Similarly for the $\alpha$ part of the earlier decomposition we have
\begin{align*}
\Pn\{ \varphi_\alpha(\hat{\boldsymbol\eta}) \} - \Pb\{ \varphi_\alpha({\boldsymbol\eta}) \} &= (\Pn-\Pb)\{ \varphi_\alpha(\hat{\boldsymbol\eta}) - \varphi_\alpha({\boldsymbol\eta}) \} + (\Pn-\Pb)  \varphi_\alpha({\boldsymbol\eta}) - \Pb\{ \varphi_\alpha(\hat{\boldsymbol\eta}) -  \varphi_\alpha({\boldsymbol\eta}) \} 
\end{align*}
The first term is again centered empirical process and is $o_\Pb(1/\sqrt{n})$ since $\varphi_\alpha$ lies in a Donsker class and $\|\varphi_\alpha(\hat{\boldsymbol\eta}) - \varphi_\alpha({\boldsymbol\eta}) \|^2=o_\Pb(1)$ by the assumptions of Theorem 3. The second term is asymptotically normal after scaling by $\sqrt{n}$, by the central limit theorem. \\

As with the $\beta$ part of the decomposition, we have
\begin{align*}
\Pb\{ \varphi_\alpha(\hat{\boldsymbol\eta}) -  \varphi_\alpha({\boldsymbol\eta}) \} &= \delta_1 \Pb\{ \varphi_{\alpha,u}(\hat{\boldsymbol\eta}) -  \varphi_{\alpha,u}({\boldsymbol\eta}) \} + (1-\delta_1) \Pb\{ \varphi_{\alpha,\ell}(\hat{\boldsymbol\eta}) -  \varphi_{\alpha,\ell}({\boldsymbol\eta}) \} .
\end{align*}
Standard iterated expectation arguments show that
$$ \Pb\{ \varphi_{\alpha,u}(\hat{\boldsymbol\eta}) -  \varphi_{\alpha,u}({\boldsymbol\eta}) \} \lesssim \| \hat\pi_1 - \pi_1 \| \Big( \max_z \| \hat\theta_{z1} - \theta_{z1} \| \Big) $$
and this is $o_\Pb(1/\sqrt{n})$ by Assumptions 3--4 of Theorem 3. As with $\beta$, for $\gamma_1 = \theta_{10} -\theta_{00}$ we have
$$ \Pb\{ \varphi_{\alpha,\ell}(\hat{\boldsymbol\eta}) -  \varphi_{\alpha,\ell}({\boldsymbol\eta}) \} \lesssim \| \hat\pi_1 - \pi_1 \| \Big( \max_z \| \hat\theta_{z0} - \theta_{z0} \| \Big) + \| \hat\gamma_1 - \gamma_1 \| \sqrt{ \Pb(\gamma_1 < | \hat\gamma_1 - \gamma_1 | ) } , $$
and this is also  $o_\Pb(1/\sqrt{n})$ by Assumptions 3--4 of Theorem 3. \\

Therefore $\Pb\{ \varphi_\alpha(\hat{\boldsymbol\eta}) -  \varphi_\alpha({\boldsymbol\eta}) \} = o_\Pb(1/\sqrt{n})$, which implies
$$ \Pn\{ \varphi_\alpha(\hat{\boldsymbol\eta}) \} - \Pb\{ \varphi_\alpha({\boldsymbol\eta}) \} = (\Pn-\Pb)  \varphi_\alpha({\boldsymbol\eta}) + o_\Pb(1/\sqrt{n}) . $$

Hence
\begin{align*}
\hat\psi - \psi &= \Pn\{ \varphi_\alpha(\hat{\boldsymbol\eta}) \}^{-1} \left[ (\Pn-\Pb) \Big\{ \varphi_\beta({\boldsymbol\eta}) -  \psi \varphi_\alpha({\boldsymbol\eta}) \Big\} \right] + o_\Pb(1/\sqrt{n})
\end{align*}
which yields the result of the theorem, after applying the continuous mapping theorem and Slutsky's theorem (noting that $\Pn\{ \varphi_\alpha(\hat{\boldsymbol\eta}) \} - \Pb\{ \varphi_\alpha({\boldsymbol\eta}) \} = o_\Pb(1)$ by the above results).

\end{document}